\documentclass[a4paper,reqno]{amsart}

\usepackage{hyperref}
\usepackage{amssymb}%
\usepackage{amsthm}
\usepackage{complexity}%

\usepackage[backend=bibtex, maxnames=9, firstinits=true, sortcites=true, doi=false, url=false]{biblatex}

\usepackage[inline]{enumitem}

\usepackage{amsmath}
\usepackage{array}
\usepackage{fixltx2e}
\usepackage{url}

\usepackage[utf8]{inputenc}

\usepackage{stmaryrd}

\usepackage[textsize=scriptsize]{todonotes}%
\usepackage[capitalize]{cleveref}%

\usetikzlibrary{automata}
\usetikzlibrary{positioning}

\renewcommand{\setminus}{\smallsetminus} %
\renewcommand{\phi}{\varphi} %
\renewcommand{\epsilon}{\varepsilon}     %
\newcommand{\eword}{\varepsilon}%

\newcommand{\egdef}{\stackrel{\mbox{\begin{scriptsize}def\end{scriptsize}}}{=}}
\newcommand{\subword}{\sqsubseteq}
\newcommand{\strictsubword}{\sqsubset}

\newcommand{\prefix}{\sqsubseteq_{\mathsf{p}}}
\newcommand{\suffix}{\sqsubseteq_{\mathsf{s}}}
\newcommand{\ttw}[1]{{#1}}
\newcommand{\tta}{\ttw{a}}
\newcommand{\ttb}{\ttw{b}}
\newcommand{\ttc}{\ttw{c}}
\newcommand{\ttd}{\ttw{d}}
\newcommand{\tte}{\ttw{e}}
\newcommand{\ttp}{\ttw{p}}

\newcommand{\ttv}{\ttw{v}}
\newcommand{\ttww}{\ttw{w}}
\newcommand{\ttsharp}{\ttw{\#}}
\newcommand{\cA}{\mathcal{A}}
\newcommand{\cB}{\mathcal{B}}
\newcommand{\Nat}{\mathbb{N}}
\newcommand{\N}{\mathbb{N}}
\newcommand{\up}{{\uparrow}}
\newcommand{\down}{{\downarrow}}
\newcommand{\LL}[1]{\up{#1}}

\newcommand{\frag}[2]{\Sigma_{#1,#2}}
\newcommand{\fragf}[2]{\Sigma'_{#1,#2}}
\newcommand{\fragpif}[2]{\Pi'_{#1,#2}}
\newcommand{\fragpi}[2]{\Pi_{#1,#2}}

\newcommand{\FOpure}[1]{\FO({#1}^*, \mathord{\subword})}
\newcommand{\FOconst}[2][XX]{\FO({#2}^*, \mathord{\subword}, \ifthenelse{\equal{#1}{XX}}{\ttw{w}_1, \mathord{\ldots}}{#1})}
\newcommand{\swstructpure}[1]{({#1}^*, \mathord{\subword})}
\newcommand{\swstructconst}[1]{({#1}^*, \mathord{\subword}, \ttw{w}_1, \mathord{\ldots})}

\DeclareMathOperator{\Ima}{Im}

\newcommand{\WeakExp}[1]{{\ComplexityFont{\Sigma}}_{#1}^{\ComplexityFont{EXP}}}
\newcommand{\autstep}[1][YYY]{\ifthenelse{\equal{#1}{YYY}}{\rightarrow}{\rightarrow_{#1}}}
\newcommand{\autsteps}[1][YYY]{\ifthenelse{\equal{#1}{YYY}}{\xrightarrow{*}}{\xrightarrow{*}_{#1}}}
\newcommand{\langof}[1]{L(#1)}
\newcommand{\assignments}[2][ZZZ]{\ifthenelse{\equal{#1}{ZZZ}}{\llbracket #2 \rrbracket}{\llbracket #2 \rrbracket_{#1}}}

\newcommand{\freevar}[1]{\mathsf{fv}(#1)}
\newcommand{\im}[1]{\mathsf{im}(#1)}

\newcommand{\direction}[2]{``\ref{#1}$\Rightarrow$\ref{#2}''}
\newcommand{\Powerset}[1]{2^{#1}}

\newtheorem{theorem}{Theorem}[section]

\newtheorem{corollary}[theorem]{Corollary}

\newtheorem{fact}[theorem]{Fact}{\bfseries}{\itshape}
\newtheorem{lemma}[theorem]{Lemma}
\newtheorem{proposition}[theorem]{Proposition}

\newlist{typelist}{enumerate}{1}
\setlist[typelist,1]{label=(\roman*)}
\crefname{typelisti}{type}{types}

\begin{document}

\title[First-Order Logic over the Subword Ordering]{Decidability, Complexity, and Expressiveness of First-Order Logic Over
the Subword Ordering}

\newcommand{\ouraffiliation}{%
Lab.\ Specification \& Verification (LSV), CNRS \& ENS Paris-Saclay, Cachan, France}

\author{Simon Halfon}
\address{\ouraffiliation}
\email{halfon@lsv.fr}

\author{Philippe Schnoebelen}
\address{\ouraffiliation}
\email{phs@lsv.fr}

\author{Georg Zetzsche}
\address{\ouraffiliation}
\thanks{The third author is supported by a fellowship within the
Postdoc-Program of the German Academic Exchange Service (DAAD) and by Labex
DigiCosme, Univ.\ Paris-Saclay, project VERICONISS}
\email{zetzsche@lsv.fr}

\maketitle

\begin{abstract}
We consider first-order logic over the subword ordering on finite
words where each word is available as a constant. Our first
result is that the $\Sigma_1$ theory is undecidable (already over
two letters). 

We investigate the decidability border by considering fragments where
all but a certain number of variables are alternation bounded, meaning
that the variable must always be quantified over languages with a
bounded number of letter alternations.  We prove that when at most two
variables are not alternation bounded, the $\Sigma_1$ fragment is
decidable, and that it becomes undecidable when three variables are
not alternation bounded.  Regarding higher quantifier alternation
depths, we prove that the $\Sigma_2$ fragment is undecidable already
for one variable without alternation bound and that when all variables
are alternation bounded, the entire first-order theory is decidable.

\end{abstract}

\section{Introduction}
\label{sec-intro}

A subsequence of a (finite) sequence $u$ is a
sequence obtained from $u$ by removing any number of elements. For
example, if $u=(\tta, \ttb, \tta, \ttb, \tta)$ then $u'=(\ttb, \ttb
,\tta)$ is a subsequence of $u$, a fact we denote with $u'\subword
u$. Other examples that work for any $u$ are $u\subword u$ (remove
nothing) and $()\subword u$.  In the rest of this paper, we shall use
the terminology from formal methods and will speak of \emph{words} and
their \emph{subwords} rather than finite sequences.

Reasoning about subwords occurs prominently in many areas of computer
science, e.g., in pattern matching (of texts, of DNA strings, etc.),
in coding theory, in theorem proving, in algorithmics, etc.  Closer to
our own motivations, the automatic verification of unreliable channel
systems and related problems involves the subword ordering or some of
its
variants~\cite{abdulla-forward-lcs,BMOSW-fac2012,HSS-lmcs,KS-msttocs}.
Our experience is that reasoning about subwords and related concepts
(e.g., shuffles of words) involves ad hoc
techniques quite unlike the standard tools that work well with
prefixes and suffixes~\cite{KNS-tcs2016}.
\\

\subsection*{The logic of subwords}
In this paper we consider the first-order logic $\FO(A^*,\subword)$ of
words over some alphabet $A=\{\tta,\ttb,\ttc,\ldots\}$ equipped with the subword relation $\subword$.
Our main objective is to understand how and when one can decide
queries formulated in this logic, or decide whether a given formula is
valid.

For example, we consider formulas like
{\makeatletter \tagsleft@true \makeatother %
\begin{gather}
\tag*{$\phi_1$:}
\forall u,u',u'':
 u\subword u'\land u'\subword u''\implies u\subword u''
\:,
\\[.3em]
\tag*{$\phi_2$:}
\exists u:\:
\tta\ttb\ttc\ttd\subword u
\land
\ttb\ttc\ttd\tte\subword u
\land
\tta\ttb\ttc\ttd\tte\not\subword u
\:,
\\[.3em]
\tag*{$\phi_3$:}
\forall u, v: \exists s:
\left(\!\!\!
\begin{array}{rl}
         & u\subword s\land v\subword s\\
\land\!\!& \forall t: u\subword t\land v\subword t \implies s\subword t
\end{array}
\!\!\!\right)
\:.
\end{gather}
}%
Here $\varphi_1$ states that the subword relation is transitive
(which it is).

More interesting is $\varphi_2$, stating that it is possible that
a word contains both $\tta\ttb\ttc\ttd$ and $\ttb\ttc\ttd\tte$ as
subwords but not $\tta\ttb\ttc\ttd\tte$. This formula is true and,
beyond knowing its validity, one is also interested in
\emph{solutions}: can we design a constraint solver that will produce
a witness, e.g., $u=\ttw{bcdeabcd}$, or more generally the set of
solutions?

Our third example, $\varphi_3$, states that words ordered by
subwords are an upper semilattice. This is a more complex formula with
$\Pi_3$ quantifier alternation. It is not valid in general (e.g.,
$\tta\ttb$ and $\ttb\tta$ have no lub) but this depends on the
alphabet $A$ at hand: $\varphi_3$ holds if $A$ is a singleton alphabet, i.e.,
$\{\tta\}^*\models\varphi_3$ but $\{\tta,\ttb\}^*\not\models\varphi_3$.

We say that formulas like $\varphi_1$ or $\varphi_3$ where
constants from $A^*$ do not appear are in the \emph{pure fragment}.
Formally, there are two logics at hand here. The \emph{pure} logic is the
logic of the purely relational structure $\swstructpure{A}$
while the \emph{extended} logic is over the expansion
$\swstructconst{A}$ where there is a constant
symbol $\ttw{w}_i$ for every word in $A^*$.

As we just illustrated with $\varphi_3$, the validity of a
formula may depend on the underlying alphabet even for the pure
fragment.  We note that this phenomenon is not limited to the
degenerate case of singleton alphabets. Indeed, observe that it is
possible to state that $u$ is a letter, i.e., is a word of length 1,
in the pure fragment:
\[
\exists z:\forall x: z\subword x
\land
(x\subword u\implies \bigl(u\subword x\lor x\subword z\bigr))
\:.
\]
Thus, even in the pure fragment, one can state that $A$ contains
$2$, $3$, \ldots, or exactly $n$ letters. Similarly one can state that $A$ is infinite by
saying that  no word contains all letters.
\\

\subsection*{State of the art}

Relatively little is known about deciding the validity of $\swstructpure{A}$
formulas and about algorithms for computing their solutions. By comparison, it
is well known that the $\Sigma_2$-theory of $\FO(A^*,\cdot)$, the
logic of strings with concatenation, is
undecidable~\cite{quine46,durnev1995},  and that its
$\Sigma_1$ fragment (aka ``word equations'') is decidable in
$\PSPACE$~\cite{plandowski2006,jez2016}. Moreover, introducing counting
predicates leads to an undecidable $\Sigma_1$ fragment~\cite{buchi88}.

Regarding the logic of subwords,  Comon and
Treinen showed undecidability for an extended logic
$\FO(A^*,\subword,p_\ttsharp)$ where $A=\{\tta,\ttb,\ttsharp\}$ has three
letters and $p_\ttsharp$ is a unary function that prepends $\ttsharp$
in front of a word, hence is a restricted form of
concatenation~\cite[Prop.~9]{comon94}.  Kuske showed that, when only
the subword predicate is allowed, the logic $\FOpure{A}$ is
undecidable and already its $\Sigma_3$ fragment is undecidable when
$|A|\geq 2$.  Kudinov \textit{et al.} considered definability in
$\swstructpure{A}$ and showed that the predicates definable in
$\swstructpure{A}$ are exactly the arithmetical
predicates\footnote{Those
that are invariant under the automorphisms of the
structure}~\cite{kudinov2010}.

Kuske's result on the $\Sigma_3$ theory leaves open the question
whether smaller fragments are decidable. Karandikar and Schnoebelen
showed that the $\Sigma_2$ theory is undecidable~\cite{KS-fosubw} and
this is tight since the $\Sigma_1$ fragment is decidable, in fact
$\NP$-complete~\cite{kuske2006,KS-fosubw}. 

Karandikar and Schnoebelen
also showed that the two-variable fragment $\FO^2(A^*,\subword)$ is
decidable~\cite{KS-fosubw} and that it has an elementary complexity
upper bound~\cite{KS-csl2016}. Decidability extends to
the logic $\FO^2(A^*,\subword,R_1,R_2,\ldots)$ where arbitrary regular languages
(monadic predicates) are allowed.
\\

\subsection*{Objectives of this paper}

We are interested in solving constraints built with the subword
ordering.  This corresponds to the $\Sigma_1$ fragment but beyond
deciding validity, we are interested in computing sets of solutions: a
formula like $\varphi_2$ can be seen as a conjunctive set of
constraints, ``$\tta\ttb\ttc\ttd\subword x \land \ttb\ttc\ttd\tte\subword
x \land \tta\ttb\ttc\ttd\tte\not\subword x$'' that define a set of words (a
set of tuples when there are several free variables).

A first difficulty is that Kuske's decidability result for the
$\Sigma_1$ fragment only applies to the pure fragment, where constants
are not allowed. That is, we know how to decide the validity of
formulas like $\varphi_1$ but not like $\varphi_2$.
However, using constants inside constraints is natural and
convenient. In particular, it makes it easy to express
piecewise testable constraints (see below), and
we would like to generalise Kuske's result to the extended logic.

We note that, in principle, the difference between the pure and the
extended logic is only superficial since, up to automorphisms,
arbitrary words can be
defined in the logic,\footnote{This is a common situation,
shared with, e.g., $\FO(A^*,\cdot)$ and $\FO(\N,<)$.}
see~\cite{kuske2006,kudinov2010,KS-fosubw}. However this requires
some universal quantification (even when defining the empty word)  that are not allowed when restricting to the
$\Sigma_1$ fragment. So this avenue is closed.
\\

\subsection*{Summary of results}

Our first result is that, \emph{when constants are allowed}, the
$\Sigma_1$ fragment of $\FOconst{A}$ is actually
undecidable. 
In fact the $\Sigma_1$ fragment of $\FOconst[W]{A}$, where \emph{a single
constant} $W\in A^*$ can be named, is undecidable unless $W$ is too
simple.
These results hold as soon as $A$ contains two distinct letters
and exhibit a sharp contrast between the pure and the extended logic.
We found this very surprising because, before hitting on undecidability,
we had already developed algorithms that solve large classes of $\Sigma_1$ constraints.
\\

Our second result identifies a key factor influencing decidability: it
turns out that free variables ranging over a  ``thin''
language like $L=\tta^+\ttb\ttc^*$, are easier to handle than
variables ranging over a ``wide''
language like $L'=(\tta+\ttb)^*$.  The key difference is that a thin language
only allows a bounded number of letter changes (in $L$ we have $\tta$'s, then
$\ttb$'s, then $\ttc$'s) while a wide language contains words with arbitrarily
many alternations between distinct letters.

These observations lead to a new descriptive complexity measure for
the formulas in $\FOconst{A}$. The associated fragments, denoted $\frag{i}{j}$
for $i,j\in\Nat$, consist of all $\Sigma_i$ formulas where $j$ variables, say
$x_1,x_2,\ldots,x_j$ can be used without any restrictions, while all the other
variables must be restricted with respect to letter alternations, say using
$x\in (\tta^*_1\tta^*_2\cdots \tta^*_n)^\ell$ for some $\ell\in \Nat$ and assuming that
$\tta_1,\ldots,\tta_n$ is a fixed enumeration of $A$.  In computer-aided
verification, such bounded quantifications occur in the analysis of bounded
context-switching protocols.

Within this classification framework, we can delineate a precise
undecidability landscape. The $\frag{1}{2}$ fragment is decidable
while $\frag{1}{3}$ is undecidable even for $|A|=2$.  The
$\frag{2}{0}$ fragment is decidable while $\frag{2}{1}$ is not.  In
fact, when all variables are alternation bounded, the entire
first-order theory is decidable.

\begin{table}
\begin{center}
\begin{tabular}{c|cccc}
$\frag{i}{j}$ & 0 & 1 & 2 & 3 \\ \hline
1             & $\NP$ & $\NP$ & in $\NEXP$ & U \\
$i\ge 2$      & $\WeakExp{i-1}$ & U & U & U
\end{tabular}
\end{center}
\caption{The cell in row $i$ and column $j$ shows the decidability/complexity
of the fragment $\frag{i}{j}$.}\label{frag:decidability}
\label{tab-res}
\end{table}

The computational complexity of all mentioned fragments is summarized
in \cref{tab-res}. Note that, in this table, $\WeakExp{n}$ denotes the
$n$-th level of the weak $\EXP$ hierarchy, which lies between $\NEXP$ and
$\EXPSPACE$~\cite{hemachandra89,gottlob95b}.
\\

Finally, we offer a series of expressiveness results showing how various
predicates like concatenation or length function can, or cannot, be defined in
the $\frag{i}{j}$ fragments.  As demonstrated in the paper, expressiveness
results are crucial to obtain hardness results.  Beyond their theoretical
interest, and since pinning down precise properties of words is not easy when
only the subword ordering is available, these results provide a welcome
intermediate language for defining more complex formulas.
\\

\subsection*{Related work}

We already mentioned works on the logic of concatenation, or the
two-variable fragment $\FO^2(A^*,\subword)$.  Because undecidability
appears so easily when reasoning about words, the focus is often on
restricted fragments, typically $\Sigma_1$, aka ``constraint
solving''. Decision methods for constraints over words have been
considered in several contexts but this usually does not include the
subsequence predicate: these works rather consider the prefix
ordering, and/or membership in a regular language, and/or functions
for taking contiguous subsequences or computing the length of
sequences, see, e.g.,~\cite{hooimeijer2012,ganesh2013,abdulla2014}.
\\

\subsection*{Outline of the paper}
We provide in \Cref{sec-basics} the basic definitions and results
necessary for our later developments.  Then we show the undecidability of
the $\Sigma_1$ fragment 
(\Cref{undecidability}) before focusing on the decidable fragments
(\Cref{complexity}). Finally, in \Cref{expressiveness}, we turn to
expressiveness questions.

\section{Subwords and their logics}
\label{sec-basics}

We consider finite words $\ttww,\ttv,...$ over a
given finite alphabet $A$ of letters like $\tta,\ttb,\ldots$.
Concatenation of words is written multiplicatively, with the
empty word $\epsilon$ as unit.  We freely use regular expressions like
$(\tta \ttb)^*+(\ttb \tta)^*$ to denote regular languages.

The length of a word $\ttww$ is written $|\ttww|$ while, for a letter $\tta\in
A$, $|\ttww|_\tta$ denotes the number of occurrences of $\tta$ in $\ttww$. The set
of all words over $A$ is written $A^*$.

A word $\ttv$ is a \emph{factor} of $\ttww$ if there exist words $\ttww_1$ and
$\ttww_2$ such that $\ttww = \ttww_1 \ttv \ttww_2$.  If furthermore $\ttww_1=\epsilon$ then
$\ttv$ is a \emph{prefix} of $\ttww$,
while if $\ttww_2=\epsilon$ then $\ttv$ is a \emph{suffix}.

\subsection*{Subwords.}

We say that a word $\ttww$ is a \emph{subword} (i.e., a subsequence) of
$\ttv$, written $\ttww\subword \ttv$, when $\ttww$ is some $\tta_1\cdots \tta_n$ and $\ttv$
can be written as $\ttv_0 \tta_1 \ttv_1 \cdots \tta_n \ttv_n$ for some
$\ttv_0,\ttv_1,\ldots,\ttv_n\in A^*$, e.g., $\epsilon\subword \ttb \ttb \tta\subword \tta \ttb \tta \ttb \tta$.
We write $\ttww\strictsubword \ttv$ for the
associated strict ordering, where $\ttww\neq \ttv$. Two words $\ttww$ and $\ttv$ are
\emph{incomparable} (with respect to the subword relation), denoted $\ttww
\perp \ttv$, if $\ttww \not \subword \ttv$ and $\ttv \not \subword \ttww$.
Factors are a special case of subwords.

With any $\ttww\in A^*$ we associate its upward closure $\LL{\ttww}$, given by
$\LL{\ttww}\egdef\{\ttv\in A^* ~|~ \ttww\subword \ttv\}$.  For example,
$\LL{\ttw{ab}}=A^*\tta A^*\ttb A^*$. The definition of $\LL{\ttww}$ involves an implicit alphabet
$A$ that will always be clear from the context.

\subsection*{Piecewise testable languages}
Piecewise testable languages (abbreviated PT) constitute a subvariety of the
languages of dot-depth one, themselves a subvariety of the star-free languages,
which are a subvariety of the regular languages~\cite{DGK-ijfcs08}.
Among the several characterizations of PT languages, the most
convenient for our purposes is the following one: $L\subseteq A^*$ is
PT if, and only if, it is a boolean combination of languages of the
form $\up \ttww$ for some $\ttww\in A^*$. Thus the PT languages are exactly
the monadic predicates that can be defined by a boolean combination of
constraints of the form $\ttww_i\subword x$ and/or $\ttww_j\not\subword x$, or
equivalently by a quantifier-free $\varphi_L(x)$ formula in the
$\FOconst{A}$ logic. For example, the solutions of $\phi_2$ (from
the introduction) form a PT language. 
In the following, we often write ``$x\in L$'', where $L$ is a given PT
language, as an abbreviation for $\phi_L(x)$, with the understanding
that this is a $\Sigma_0$ formula.

\subsection*{Logic of subwords}

Let $V$ be the set of variables with typical elements 
$x,y,\ldots,u,v,\ldots$. 
 For a
first-order logic formula $\varphi$ over a structure with domain $D$, we denote
by $\assignments{\varphi}\subseteq D^V$ the set of satisfying
assignments, with typical elements $\alpha,\beta,\ldots$.
If $\varphi$ has only one free variable, say $x$, and there is no danger of
confusion, we sometimes write $\assignments{\varphi}$ to mean $\{\alpha(x) \mid \alpha\in\assignments{\varphi}\}$.
Moreover, $\freevar{\varphi}$ denotes the set of free variables in $\varphi$.

By $\FOpure{A}$, we denote the first-order logic over the structure
$\swstructpure{A}$. In contrast, $\FOconst{A}$ is the first-order logic over the
structure $\swstructconst{A}$, where for each word $\ttww\in A^*$, the signature
provides a constant symbol. In the case of $\FOconst{A}$ and $\FOpure{A}$,
assignments are members of $(A^*)^V$. We will sometimes write $\ttww$ to denote the
assignment that maps every variable to the word $\ttww\in A^*$.  Moreover,
$(x\mapsto \ttww)$ denotes the assignment in $(A^*)^{\{x\}}$ that maps $x$ to $\ttww$.

\subsection*{Bounding alternations}
We define a fragment of first-order logic over the relational structure $\swstructconst{A}$. Let $A=\{\tta_1,\ldots,\tta_n\}$.  The starting point for
introducing the fragments $\frag{i}{j}$ is the observation that if
every variable in a sentence $\varphi$ is introduced by a restricted quantifier
of the form $\exists x\in(\tta_1^*\cdots \tta_n^*)^\ell$ or $\forall
x\in(\tta_1^*\cdots \tta_n^*)^\ell$ for
some $\ell\in\N$, then one can reduce the truth problem 
of $\varphi$ to Presburger arithmetic. Note that the language $(\tta_1^*\cdots
\tta_n^*)^\ell$ is PT, implying that such restrictions, which we
call \emph{alternation bounds}, can be imposed within $\FOconst{A}$
and without any additional quantifiers. This raises
the question of how many variables without alternation bound one can allow
without losing decidability. 

In essence $\frag{i}{j}$ contains all formulas in the $\Sigma_i$
fragment with $j$
variables without alternation bound. A formalization of this \emph{for sentences}
could just be a syntactic restriction: Every quantifier for all but at most $j$
variables must be relative to some $(\tta_1^*\cdots \tta_n^*)^\ell$. However, this
would not restrict free variables, which we need in order to build complex $\frag{i}{j}$
formulas from predicates defined in $\frag{i}{j}$.

Formally, a \emph{formula with alternation
bounds} consists of a formula $\varphi$ of $\FOconst{A}$ and a function
$\ell\colon V\to\N\cup\{\infty\}$, which specifies the alternation bounds. This
means, the semantics $\assignments{(\varphi,\ell)}$ of $(\varphi,\ell)$ is
defined as $\assignments{\tilde{\varphi}}$, where $\tilde{\varphi}$ is defined
as follows.  First, we replace every quantifier $\mathcal{Q}x$
($\mathcal{Q}\in\{\exists,\forall\}$) in $\varphi$ by the relativized $\mathcal{Q}x\in (\tta_1^*\cdots
\tta_n^*)^{\ell(x)}$. Then we add the conjunction 
$\bigwedge_{x\in\freevar{\varphi}, \ell(x)<\infty} x\in (\tta_1^*\cdots
\tta_n^*)^{\ell(x)}$ for the free variables.

The fragment $\frag{i}{j}$ consists of those formulas with alternation bounds
$(\varphi,\ell)$ where $\varphi$ belongs to the $\Sigma_i$ fragment of
$\FOconst{A}$ and has at most $j$ variables $x\in V$ with $\ell(x)=\infty$.
We will always represent a formula in $\frag{i}{j}$ by its $\Sigma_i$ formula
and the function $\ell$ will be clear from the context.
Variables $x\in V$ with $\ell(x)<\infty$ will be called \emph{alternation bounded},
the others \emph{alternation unbounded}.
In order to permit a polynomial translation into an equivalent formula in
ordinary $\FOconst{A}$, the alternation bounds are always encoded in unary.
The fragment $\fragpi{i}{j}$ is defined similarly, with $\Pi_i$ instead of $\Sigma_i$.

Sometimes we define predicates that are satisfied for words with unbounded
alternations (such as ``$u\in \{\tta,\ttb\}^*$'' when $A=\{\tta,\ttb,\ttc\}$), but want to use
the corresponding formula in a context where the variables are alternation
bounded (``$u\in\{\tta,\ttb\}^* \wedge \tta \ttb\not\subword u$''). In that situation, we
want to record the number of alternation unbounded variables we need for the
definition, \emph{disregarding the free variables}. Hence, $\fragf{i}{j}$
denotes those formulas with alternation bound in $\Sigma_i$, where there are at
most $j$ \emph{quantified variables} without alternation bound.  The semantics
is defined as for $\frag{i}{j}$.  The fragment $\fragpif{i}{j}$ is defined with
$\Pi_i$ instead of $\Sigma_i$.

\section{Undecidability}\label{undecidability}

\subsection{The $\frag{1}{3}$ fragment}
We begin with our main result, the undecidability of the $\Sigma_1$ theory of
$\FOconst{A}$ for $|A|\ge 2$. In fact, we will even prove undecidability for
the $\frag{1}{3}$ fragment.  We need a few ingredients. A word $\ttww\in A^+$ is
called \emph{primitive} if there is no $\ttv\in A^+$, $|\ttv|<|\ttww|$, with $\ttww\in \ttv^*$.
The following is a well-known basic fact from word combinatorics (see e.g.
\cite[Exercise 2.5]{berstel79})
\begin{fact}
\label{cor-prim}
If $\ttp \in A^+$ is primitive, then $\ttp\ttww=\ttww \ttp$ is equivalent to $\ttww\in \ttp^*$.
\end{fact}

We also use the following version of the fact that Diophantine sets
are precisely the recursively enumerable sets~\cite{Mat93}.
\begin{theorem}\label{diophantine}
Let $S\subseteq \N$ be a recursively enumerable set. Then there is a finite set
of variables  $\{x_0,\ldots,x_m\}$ and a finite set $E$ of equations, each of
the form
\begin{align*}
x_i=x_j+x_k && x_i=x_j\cdot x_k && x_i=1
\end{align*}
with $i,j,k\in [0,m]$, such that 
\[S=\{y_0 \in\N \mid \exists y_1,\ldots,y_m\in\N\colon \text{$(y_0,\ldots,y_m)$ satisfies $E$}\}. \] 
\end{theorem}

We are now ready to prove our main result.
\begin{theorem}\label{subword:undecidable}
Let $|A|\ge 2$ and $a\in A$. For each recursively enumerable set
$S\subseteq\N$, there is a $\frag{1}{3}$ formula $\varphi$ over the structure
$\FOconst{A}$ with one free variable such that $\assignments{\varphi}=\{a^k
\mid k\in S\}$. 
\end{theorem}
\begin{proof}
We show how to express some basic properties of words and combine these to
build more complex predicates, all the time keeping track of what fragments are
involved.  Here, we always use $u,v,w$ as the free variables of the formula we
currently construct.

Recall that for every PT language $L\subseteq A^*$, we can express ``$u\in
L$'' in $\fragf{0}{0}$: we will use this silently, mainly for
languages of the form $r a^* s$ where $a$ is a letter and $r,s$ are
two words.\footnote{That any language of the form $ra^*s$ is PT is
easy to prove, e.g., using the characterization of~\cite{klima2013}.}
Note also that, since ``$u\in(a+b)^*\,$'' can be expressed in $\fragf{0}{0}$ for $a,b\in A$,
it suffices to prove the theorem in the case $|A|=2$.
\begin{enumerate}
  
\item We can express ``$|u|_a<|v|_a$'' in $\fragf{1}{0}$:
  \[
  \exists x\in a^*\colon x\subword v \wedge x\not\subword u.
  \]
  
\item We can express ``$\exists n\colon u=a^n \wedge v=a^{n-1}b$'' in $\frag{1}{0}$.
  Clearly, it suffices to show that we can express
  ``$\exists n\ge 2\colon u=a^n \wedge v=a^{n-1}b$''.
  Consider the formula:
  \begin{align*}
  u\in a a a^* \wedge v\in a^*b\wedge
  \exists x\in a^*b a a\colon
   |v|_a<|u|_a \wedge v\not\subword x \wedge u\subword x.
  \end{align*}
  Suppose the formula is satisfied with $u=a^n$, $x=a^\ell b a a$ and $v=a^m b$.
  Then $|v|_a<|u|_a$ implies $m<n$. By $v\not\subword x$, we have $\ell<m$ and
  thus $\ell<m<n$, hence $\ell+2\le n$.  On the other hand, $u\subword x$ implies
  $n\le \ell+2$ and thus $n=\ell+2$ and $m=n-1$.

  Conversely, if $u=a^n$ and $v=a^{n-1}b$ for some $n\geq 2$, then the formula is satisfied with
  $x=a^{n-2}b a a$.

\item\label{undec2:last} We can express ``$u,v\in (a+b)^*b \wedge |u|_a=|v|_a$'' in $\fragf{1}{0}$:
  \begin{align*}
  u,v\in (a+b)^* \\ ~\wedge~
  \exists x \in a^*\colon \exists y \in a^*b \colon 
  \left[ \exists n\colon x=a^n \wedge y=a^{n-1}b \right] \\ ~\wedge~
   y\subword u~\wedge~ y\subword v ~\wedge~ x\not\subword u~\wedge~ x\not\subword v.
  \end{align*}
  Suppose the formula is satisfied. Then $a^{n-1}b\subword u$ and
  $a^n\not\subword u$ imply $|u|_a=n-1$. Moreover, if $u$ ended in $a$, then
  $a^{n-1}b\subword u$ would entail $a^n\subword u$, which is not the case. Since
  $|u|\ge 1$, we therefore have $u\in\{a,b\}^*b$.  By symmetry, we have
  $|v|_a=n-1$ and $v\in\{a,b\}^*b$. Hence, $|u|_a=n-1=|v|_a$.

  If $u,v\in\{a,b\}^*b$ with $|u|_a=|v|_a$, then the formula is satisfied with $n=|u|_a+1$.
  
\item We can express ``$\exists n\colon u=aaba^{n}b \wedge v=aba^{n+1}b\wedge w=ba^{n+2}b$'' in $\frag{1}{0}$:
  \begin{align*}
  u\in aaba^*b \wedge v\in aba^*b\wedge w\in ba^*b \\ ~\wedge~
  \left[u,v,w\in \{a,b\}^*b \wedge |u|_a=|v|_a=|w|_a\right].
  \end{align*}
\item\label{undec2:bab+1} We can express ``$\exists n\colon u=ba^nb\wedge v=ba^{n+1}b$'' in $\frag{1}{0}$. It suffices to show that we can
express ``$\exists n\ge 1\colon u=ba^nb\wedge v=ba^{n+1}b$''. Consider the formula:
\begin{align*}
  \exists x \in aaba^*b, y \in aba^*b, z\in ba^*b \colon \\
  \left[ \exists m\colon x=aaba^mb \wedge y=aba^{m+1}b \wedge z=ba^{m+2}b \right] \\
~\wedge~ u,v\in ba^*b ~\wedge~ u\subword y ~\wedge~ u\not\subword x ~\wedge~ v\subword z~\wedge~ v\not\subword y. 
\end{align*}
Suppose the formula is satisfied for $u=b a^kb$ and $v=b a^\ell b$. Then
$u\subword y$ and $u\not\subword x$ imply $k\le m+1$ and $k>m$, hence $k=m+1$.
Moreover, $v\subword z$ and $v\not\subword y$ imply $\ell\le m+2$ and
$\ell>m+1$, hence $\ell=m+2$. Hence, with $n=m+1$ we have $u=b a^nb$ and
$v=b a^{n+1}b$ and $n\ge 1$.

Conversely, if $u=b a^nb$ and $v=b a^{n+1}b$ for some  $n\ge 1$, then the formula is satisfied with $m=n-1$.

\item\label{undec2:anan+1} We can express ``$\exists n\colon u=a^n \wedge v=a^{n+1}$'' in $\frag{1}{0}$. For this, it suffices to
express ``$\exists n\ge 1\colon u=a^n \wedge v=a^{n+1}$''.
As in \cref{undec2:bab+1}, one verifies correctness of the following:
\begin{align*}
  \exists x,y,z\colon
  \left[ \exists m\colon x=b a^mb \wedge y=b a^{m+1}b \wedge z=b a^{m+2}b \right] \\
\wedge ~u,v\in a^*~\wedge~ u\subword y~\wedge~ u\not\subword x ~\wedge~ v\subword z ~\wedge~ v\not\subword y.
\end{align*}

\item We can express ``$v=a^{|u|_a}$'' in $\fragf{1}{0}$:
  \[
  \exists x \in a^*\colon
  \left[ \exists n\colon v=a^n \wedge x=a^{n+1} \right]
  \wedge v\subword u \wedge x\not\subword u.
  \]
  
\item\label{undec2:counting} We can express ``$|u|_a=|v|_a$'' in $\fragf{1}{0}$:
  \[
  \exists x\colon x=a^{|u|_a} \wedge x=a^{|v|_a}.
  \]
  
\item\label{undec2:append-left} For $a\ne b$, we can express ``$u\in a^* \wedge v=bu$'' in $\frag{1}{0}$:
  \[
  u\in a^* \wedge v\in b a^* \wedge |v|_a=|u|_a.
  \]
  
\item\label{undec2:append-right} For $a\ne b$, we can express ``$u\in a^* \wedge v=ub$'' in $\frag{1}{0}$:
  \[
  u\in a^* \wedge v\in a^*b \wedge |v|_a=|u|_a.
  \]
  
\item\label{undec2:add} We can express ``$|w|_a=|u|_a+|v|_a$'' for any $a\in A$ in $\fragf{1}{0}$.
Let $b\in A\setminus \{a\}$:
\begin{align}
\exists x,y\in a^*\colon\exists z\in a^*b a^*\colon x=a^{|u|_a}\wedge y=a^{|v|_a} \nonumber \\
\wedge \ xb\subword z \wedge xab \not\subword z \wedge by\subword z\wedge bya\not\subword z \label{undec2:exact:bounded}\\
\wedge \ |w|_a=|z|_a
\end{align}
Note that we can define $xb$, $(xa)b$ and $b(ya)$ thanks to \cref{undec2:append-left,undec2:append-right,undec2:anan+1}.
The constraints in \cref{undec2:exact:bounded} enforce that $z=xby$ and hence $|z|_a=|x|_a+|y|_a=|u|_a+|v|_a$.

\item For $k,n_0,\ldots,n_k\in\N$, $a\ne b$, 
let $r_a(a^{n_0}b a^{n_1}\cdots
b a^{n_k})=n_k$, which defines a function $r_a\colon \{a,b\}^*\to\N$.
We can express ``$v=a^{r_a(u)}$'' in $\fragf{1}{0}$:
\begin{align*}
v\in a^* \wedge~\exists x\in b^*a^* \colon \exists y\in b^*a^* \colon \\
|x|_b=|y|_b=|u|_b ~\wedge~ |y|_a=|x|_a+1 \\
\wedge ~x\subword u~\wedge~ y\not\subword u~\wedge~ |v|_a=|x|_a.
\end{align*}
Note that $|x|_b=|y|_b=|u|_b$ can be expressed according to
\cref{undec2:counting} and $|y|_a=|x|_a+1$ can be expressed thanks to
\cref{undec2:add}.
Write $u=a^{n_0}b a^{n_1}\cdots b a^{n_k}$.

Suppose the formula is satisfied. Then $|x|_b=|y|_b=|u|_b$ and $|y|_a=|x|_a+1$
imply that $x=b^ka^m$ and $y=b^ka^{m+1}$ for some $m\in\N$. Moreover,
$x\subword u$ implies $m\le n_k$ and $y\not\subword u$ implies $m+1>n_k$, thus
$m=n_k$. Thus, $|v|_a=|x|_a$ entails $|v|_a=n_k$.

Conversely, if $v=a^{n_k}$, then the formula is satisfied with $x=b^ka^{n_k}$ and $y=b^ka^{n_k+1}$.

\item \label{undec2:conc:unary} For $a\in A$, we can express ``$v\in a^* \wedge w=uv$'' in $\fragf{1}{0}$. Let $b\ne a$ and consider the formula
\begin{align}
&v\in a^*~\wedge \nonumber \\
\wedge~~&\exists x\in a^*~ \colon\exists y \in a^* ~\colon x=r_a(u) ~\wedge~ y=r_a(w) \nonumber \\
\wedge ~~& |w|_b=|u|_b~\wedge~ u\subword w\label{undec2:conc:leq} \\
\wedge ~~& |y|_a=|x|_a+|v|_a ~\wedge~ |w|_a=|u|_a+|v|_a \label{undec2:conc:eq}
\end{align}

To show correctness, suppose the formula is satisfied with 
$u=a^{n_0}b a^{n_1}\cdots b a^{n_k}$ and $w=a^{m_0}b a^{m_1}\cdots b a^{m_\ell}$.
The conditions in \cref{undec2:conc:leq}
imply that $k=\ell$ and $w=a^{m_0}b a^{m_1}\cdots b a^{m_k}$ and $n_i\le m_i$ for $i\in[0,k]$.
The conditions in \cref{undec2:conc:eq} then entail $m_k=n_k+|v|_a$ and
$\sum_{i=0}^k m_i=\sum_{i=0}^k n_i+|v|_a$, which together is only possible if
$m_i=n_i$ for $i\in[0,k-1]$. This means we have $w=uv$.
The converse is clear. %

\item\label{undec2:prefix} We can express ``$u$ is prefix of $v$'' in $\frag{1}{3}$:
\[ \bigwedge_{a\in A} \exists x\colon \exists y\in a^*\colon x=uy ~\wedge~ x\subword v ~\wedge~ |x|_a=|v|_a. \]
Suppose the formula is satisfied. Then $uy\subword v$ for some $y\in A^*$,
which implies $u\subword v$.  Let $p$ be the shortest prefix of $v$ with
$u\subword p$. Observe that whenever $uw\subword v$, we also have $pw\subword
v$, because the leftmost embedding of $uw$ in $v$ has to match up $u$ with $p$.
Now towards a contradiction,
assume $|p|>|u|$. Then there is some $a\in A$ with $|p|_a>|u|_a$. The formula
tells us that for some $m\in\N$, we have $ua^m\subword v$ and $|u|_a+m=|v|_a$.
Our observation yields $pa^m\subword v$, and hence
$|v|_a \ge |p|_a + m > |u|_a+ m = |v|_a$, a contradiction. The converse is clear.

\item\label{undec2:conc} We can express ``$w=uv$'' in $\frag{1}{3}$: Since expressibility is
preserved by mirroring, we can express prefix and suffix by
\cref{undec2:prefix}. Let $\prefix$ and $\suffix$ denote the prefix and suffix
relation, respectively. We can use the formula
\[ u\prefix w ~\wedge~ v\suffix w ~\wedge~ \bigwedge_{a\in A} |w|_a=|u|_a+|v|_a. \]

\item For $a,b\in A$, $a\ne b$, we can express ``$u\in (a b)^*$'' in
$\frag{1}{3}$: By \cref{undec2:conc}, we can use the formula
$\exists v\colon v=uab \wedge v=abu,$
which, according to \cref{cor-prim}, is equivalent to $u\in (a b)^*$.

\item For $a,b\in A$, $a\ne b$, we can express ``$|u|_a=|v|_b$'' in $\frag{1}{3}$ by using 
 \[ 
  \exists x\in (a b)^*\colon |u|_a=|x|_a \wedge |v|_b=|x|_b.
  \]

\item\label{undec2:mult} We can express ``$\exists m,n\colon u=a^n \wedge v=a^m \wedge w=a^{m\cdot n}$'' in $\frag{1}{3}$:
  \begin{align*}
    u,v,w\in a^* \\
   ~\wedge~ \exists x \colon [\exists y,z\colon~ y=bu ~\wedge~ z=yx ~\wedge~ z=xy] \\
    ~\wedge~ |x|_b = |v|_a ~\wedge~ |w|_a=|x|_a.
  \end{align*}
The conditions in brackets require $(bu)x=x(bu)$. Since $bu\in ba^*$
is primitive, this is equivalent to $x\in (bu)^*$ (\emph{cf} \cref{cor-prim}).

\item We use the fact that every recursively enumerable set of natural numbers
is Diophantine.  Applying \cref{diophantine} to $S$ yields a finite set $E$ of
equations over the variables $\{x_0,\ldots,x_m\}$.  The formula $\varphi$ is of
the form
\[ \varphi \equiv \exists x_1, x_2, \dots, x_m \in a^*\colon \psi, \]
where $\psi$ is a conjunction of the following $\frag{1}{3}$ formulas. For each equation $x_i=1$, we add $x_i=a$.
For each equation $x_i=x_j+x_k$, we add a formula expressing $|x_i|_a=|x_j|_a + |x_k|_a$.
For each equation $x_i=x_j\cdot x_k$, we add a formula expressing $x_i=a^{|x_j|\cdot|x_k|}$.
Then we clearly have $\assignments{\varphi}=\{a^k \mid k\in S\}$.\qed
\end{enumerate}
\renewcommand{\qedsymbol}{}
\end{proof}

As an immediate consequence, one sees that the truth problem is also
undecidable for the $\Sigma_1$ fragment of the logic of subwords
without constants but enriched with predicates like ``$|u|_a=2$'' for
counting letter occurrences.

It can even be shown that there is a fixed word $W\in\{a,b\}^*$ such
that the truth problem of $\frag{1}{3}$ over $\FOconst[W]{\{a,b\}}$ is
undecidable.  
In order to show undecidability with a single constant, we will need
the fact that each word of length at least $3$ is determined by its
length and its strict subwords. For two words $u,v\in A^*$, we write
$u\sim_n v$ if $\down\{u\}\cap A^{\le n}=\down\{v\}\cap A^{\le n}$.
\begin{lemma}[\cite{KS-fosubw}]\label{shorter-subwords}
Let $n\ge 2$ and $|u|=|v|=n+1$. Then $u=v$ if and only if $u\sim_n v$.
\end{lemma}

\begin{theorem}\label{undecidable:single}
  There is a word $W\in\{a,b\}^*$ such that for every recursively
  enumerable set $S\subseteq\N$, there is a $\frag{1}{3}$-formula
  $\tau$ over the structure $\FOconst[W]{\{a,b\}}$ such that
  \[ \assignments{\tau}=\{a^k \mid k\in S\}. \]
  In particular, the
  truth problem for the $\frag{1}{3}$ fragment over
  $\FOconst[W]{\{a,b\}}$ is undecidable.
\end{theorem}
\begin{proof}
  In the proof of \cref{subword:undecidable}, we have constructed $\frag{1}{3}$ formulas
  over $\FOconst{A}$ expressing successor, addition, and
  multiplication, more precisely: expressing ``$\exists n\ge 0\colon
  u=a^n\wedge v=a^{n+1}$'' and ``$\exists m,n\ge 0\colon u=a^m\wedge
  v=a^n\wedge w=a^{m+n}$'' and ``$\exists m,n\ge 0\colon u=a^n\wedge
  v=a^m\wedge w=a^{m\cdot n}$''. Let $W_1,\ldots,W_r\in\{a,b\}^*$ be
  the constants occurring in these three $\frag{1}{3}$ formulas, plus
  $\varepsilon$. Let $m$ the maximal length of any of these words, and
  let $W=a^{m+1}b^{m+2}$.

  Let $S\subseteq\N$ be recursively enumerable. Then, according to
  \cref{diophantine} and by the choice of $W_1,\ldots,W_r$, there is a
  $\frag{1}{3}$ formula $\varphi$ that only uses constants from
  $W_1,\ldots,W_r$ and with $\assignments{\varphi}=\{a^k \mid k\in S\}$.
  We shall prove that using $W$, we can define all the words
  $W_1,\ldots,W_r$. Consider the formula
  \begin{align*}
        &\exists x_0,y_0,\ldots,x_{2m+3},y_{2m+3}\colon && x_0\strictsubword\cdots\strictsubword x_{2m+3}\subword W\wedge y_0\strictsubword \cdots \strictsubword y_{2m+3}\subword W  \wedge \\
        & \exists x'_0,\ldots,x'_{m+1}\colon && x'_0\strictsubword\cdots\strictsubword x'_{m+1}\subword W~\wedge \\
        & \exists y'_0,\ldots,y'_{m+2}\colon && y'_0\strictsubword\cdots\strictsubword y'_{m+2}\subword W~\wedge \\
        & && x_1\ne y_1 ~\wedge~ x_1\not\subword y'_{m+2} ~\wedge~ y_1\not\subword x'_{m+1} ~\wedge\\
        &\exists z_{01}\colon && x'_1\subword z_{01} ~\wedge~ x'_2\not\subword z_{01} ~\wedge~ y'_1\subword z_{01} ~ \wedge~y'_2\not\subword z_{01} ~\wedge \\
        &\exists z_{10}\colon && x'_1\subword z_{10} ~\wedge~ x'_2\not\subword z_{10} ~\wedge~ y'_1\subword z_{10} ~ \wedge~y'_2\not\subword z_{10} ~\wedge\\
        & && z_{01}\subword W ~\wedge~ z_{01} \ne z_{10}
      \end{align*}
      If it is satisfied, then $|x_i|=|y_i|=i$ for $i\in[0,2m+3]$ and
      since $x_1\ne y_1$, we have $\{x_1,y_1\}=\{a,b\}$. Since
      $x_1\not\subword y'_{m+2}$ we get $y'_i\in y_1^*$ and thus
      $y'_i=y_1^i$ for $i\in[0,m+2]$, which is only possible with
      $y_1=b$. This implies $x_1=a$.  In particular, we get
      $\{z_{01},z_{10}\}=\{ab, ba\}$. Since $z_{01}\subword W$, we
      have $z_{01}=ab$ and $z_{10}=ba$.  On the other hand, if $|x_i|=|y_i|=i$ for $i\in[0,2m+3]$,
      $x_1=a$, $y_1=b$,
      $x'_i=a^i$, $y'_j=b^j$ for $i\in[0,m+1]$, $j\in[0,m+2]$, $z_{01}=ab$, and $z_{10}=ba$,
      then the formula is clearly satisfied.

      Hence, we can already define all words of length at most $2$ and
      all words $a^i$ and $b^i$ for $i\in[0,m+1]$. This lets us define
      other predicates.
      \begin{enumerate}
      \item For each $0\le \ell\le m$, we can express
        ``$|u|_a=\ell$'' using the formula
        \[ a^\ell\subword u ~\wedge~a^{\ell+1}\not\subword u \]
        Note that since $\ell+1\le m+1$, we can already define
        $a^{\ell+1}$. The same way, we can express
        ``$|u|_b=\ell$''.
      \item For each $0\le \ell\le m$, we can express ``$|u|=\ell$''
        using the formula 
        \[ \bigvee_{i+j=\ell} |u|_a=i
        \wedge |u|_b=j. \]
      \item For each word $w\in A^{\le m}$, $|w|>2$, we can define
        $w$. We proceed by induction. For $|w|\le 2$, we
        can already define $w$. Thus, suppose we can define
        every $v\in A^{\le n}$ and let $w\in
        A^{n+1}$ with $n+1\le m$. Consider the formula
        \[ |u|=n+1 ~\wedge~\bigwedge_{v\in A^{\le n},~v\subword w}
        v\subword u ~\wedge~ \bigwedge_{v\in A^{\le
            n},~v\not\subword w} v\not\subword u. \]
        Clearly, if $u=w$, then the formula is
        satisfied. On the other hand, suppose the formula is
        satisfied.  It expresses that $\down\{u\}\cap A^{\le n}=\down\{w\}\cap A^{\le n}$.
        According to \cref{shorter-subwords}, this
        implies $u=w$.
      \end{enumerate}
      Note that all the variables we introduced to define the words in
      $A^{\le m}$ carry words of length at most $2m+3$, meaning that
      we may assume that they are alternation bounded.  Thus, we can
      define all words in $A^{\le m}$ using forumlas in $\frag{1}{0}$.
      Therefore, we can turn $\varphi$ into a $\frag{1}{3}$ formula
      $\tau$ that contains $W$ as its only constant and satisfies
      $\assignments{\tau}=\assignments{\varphi}$.

      It remains to show the second statement of the
      \lcnamecref{undecidable:single}. Let $S\subseteq\N$ be
      recursively enumerable but undecidable and let $k\in\N$ be
      given. We choose the formula $\varphi$ as above, but we
      modify it as follows. Let $\varphi_0\equiv\varphi$, and for
      $i\in[1,k]$, let $\varphi_i$ express
      \[ \exists y\colon \varphi_{i-1}(y)\wedge \exists n\colon x=a^n\wedge y=a^{n+1}. \] 
      Finally, let $\varphi_{k+1}$ be the formula $\exists x\colon
      \varphi_k(x)\wedge x=\varepsilon$.  Note that by the choice of
      $W_1,\ldots,W_r$, we may assume that $\varphi_{k+1}$ contains
      only the constants $W_1,\ldots,W_r$. Note that $\varphi_{k+1}$
      has no free variables and is true if and only if $k\in S$.  Now
      $\tau_{k+1}$ is obtained from $\varphi_{k+1}$ just as $\tau$ is
      obtained from $\varphi$. It follows as above that $\tau_{k+1}$
      is true if and only if $k\in S$. This proves the second
      statement of the \lcnamecref{undecidable:single}.
\end{proof}

Here $W$ must be complex enough: For instance, the $\Sigma_1$ fragment of
$\FOconst[\epsilon]{\{a,b\}}$ and of $\FOconst[a]{\{a,b\}}$, respectively, is decidable.
\begin{theorem}\label{decidable-short}
  The $\Sigma_1$-fragment of $\FOconst[\epsilon,a]{\{a,b\}}$ is decidable.
\end{theorem}
\begin{proof}
  We may assume that the input formula is of the form
  $\varphi\equiv \exists x_1,\ldots,x_n\colon \psi$, where $\psi$ is a conjunction
  of literals of the following forms:
  \begin{align*} 
    c\subword x && c\not\subword x && x\subword c     && x\not\subword c \\
                && x\subword y     && x\not\subword y &&
              \end{align*}
   where $c\in\{\varepsilon, a\}$ and $x\in X=\{x_1,\ldots,x_n\}$. For each literal $x\subword c$, we can
guess whether $x=\epsilon$ or $x=a$ and hence assume that these do not occur. Literals of the form $\epsilon\subword x$ are always satisfied, whereas $\epsilon\not\subword x$ is never satisfied. Hence, without loss of generality, these do not occur either and we may assume that all literals are of the form
\begin{align*}
a\subword x &&  a\not\subword x  &&   x\not\subword\epsilon &&   x\not\subword a \\
            && x\subword y && x\not\subword y. 
\end{align*}
Moreover $x\not\subword \epsilon$ is equivalent to $x\ne\epsilon$ and the literal $a\not\subword x$ is equivalent to $x\in b^*$. We can therefore assume that all literals are of the form
\begin{align*}
a\subword x &&  x\in b^*  &&   x\ne\epsilon &&   x\not\subword a \\
            && x\subword y && x\not\subword y 
\end{align*}
Let $L\subseteq X$ be the set of those variables for
which we have a $x\in b^*$ literal. Clearly, $x\in b^*$ and $x\ne\epsilon$ together
mean $x\in b^+$. In the same way, $x\in b^*$ and $x\not\subword a$ together mean $x\in b^+$.
Furthermore, $a\subword x$ and $x\in b^*$ are mutually exclusive. Hence, we can rewrite our constraint system as follows: 
\begin{itemize}
\item For each $x\in L$, we have either a constraint $x\in b^*$ or $x\in b^+$. 
\item For each $x\in X\setminus L$, we have a set of constraints of the form $a\subword x$, $x\ne\epsilon$, or $x\not\subword a$.
\item We have constraints of the form $x\subword y$ and $x\not\subword y$.
\end{itemize}
As a final reformulation step, note that every $u\in\{a,b\}^*$
satisfies either $u\in b^*$ or $a\subword u$.  Therefore, we may
assume that for every $x\in X$, we have either a constraint $x\in b^*$
or $x\in b^+$ (and hence $x\in L$) or $a\subword x$. Notice that if we
already have $a\subword x$, then $x\ne\epsilon$ is redundant. Thus, we
have the following constraints:
\begin{enumerate}
\item For each $x\in L$, we have either $x\in b^*$ or $x\in b^+$.
\item For each $x\in X\setminus L$, we have $a\subword x$ and possibly $x\not\subword a$.
\item A set of constraints of the form $x\subword y$ or $x\not\subword y$.
\end{enumerate}
We say that a partial order $(X,\le)$ is \emph{compatible} if 
\begin{enumerate}
\item $L$ is downward closed and linearly ordered,
\item for each constraint $x\subword y$ ($x\not\subword y$), we have $x\le y$ ($x\not\le y$).
\end{enumerate}
We claim that $\varphi$ is satisfied in
$\FOconst[\epsilon,a]{\{a,b\}}$ if and only if there is a compatible
partial order on $X$. Since the latter is clearly decidable, this
implies the \lcnamecref{decidable-short}.

Of course, if $\varphi$ is satisfied, then the subword ordering
induces a compatible partial order on $X$. So let us prove the
converse and suppose $(X,\le)$ is a compatible partial order and let
$P=X\setminus L$. Then, $(P,\le)$ is a partial order and we can find
some $m\ge 0$ such that $(P,\le)$ embeds into the lattice $\{0,1\}^m$
of $m$-tuples over $\{0,1\}$ with componentwise comparison. Consider
such an embedding with $m\ge 2$.  This embedding allows us to assign
to each $x\in P$ a word $u_x\in a\beta_1\cdots a\beta_m$, where
$\beta_1,\ldots,\beta_m\in\{\epsilon,b\}$ such that $x\le y$ if and
only if $u_x\subword u_y$.

Now write $L=\{\ell_1,\ldots,\ell_k\}$ with $\ell_1\le\cdots \le \ell_k$. We now define a function
$f\colon X\to \{0,\ldots,k\}$. Note
that for each $x\in P$, the set $\down\{x\}\cap L$ is a downward
closed subset of $L$ and hence of the form $\{\ell_1,\ldots,\ell_i\}$ for
some $i\ge 0$. In this case, set $f(x)=i$.  Moreover, let
\[ P_i = \{x\in P \mid \down\{x\}\cap L=\{\ell_1,\ldots,\ell_i\}\}. \] 
This
allows us to construct an assignment of words $v_x$ to variables
$x$. Let us explain the intuition.  In order ensure that
$v_{\ell_1}\not\subword v_x$ for all $x\in P_0$, we let $v_x=u_x$ and
notice that then, the words $v_x$ all contain at most $m$-many
$b$'s. Hence, we set $v_{\ell_i}=b^{m+1}$. Now, we have to make sure that
the words for $v_x$ with $x\in P_1$ all contain $v_{\ell_1}$ as a
subword, so we pad the words $u_x$ with $b$'s on the left: We set
$v_x=b^{m+1}u_x$. Now, in turn, we need to make sure that $v_{\ell_2}$
contains more than $(m+1)+m$-many $b$'s, leading to
$v_{\ell_2}=b^{2(m+1)}$, and so on. Thus, we set:
\begin{align*} v_{\ell_i} = b^{i(m+1)} && v_x = b^{f(x)\cdot (m+1)} u_x \end{align*}
for $i\in\{1,\ldots,k\}$ and $x\in P$. Let us show that this
assignment sastisfies our constraint system. 
\begin{itemize}
\item Consider a constraint $x\subword y$. We have $x\le y$. 
\begin{itemize}
\item If $x,y\in P$, then $\down\{x\}\cap L\subseteq\down\{y\}\cap L$
  and thus $f(x)\le f(y)$. Since also $u_x\subword u_y$, we have
  $v_x\subword v_y$.
\item If $y\in L$, then also $x\in L$ (since $L$ is downward closed) and thus clearly $v_x\subword v_y$.
\item If $y\in P$ and $x\in L$. Suppose $x=y_i$ and $f(y)=j$. By definition of $f$, we have $i\le j$
and hence $v_{x}=v_{\ell_i}=b^{i(m+1)}\subword b^{j(m+1)}u_y=v_y$.
\end{itemize}
\item Consider a constraint $x\not\subword y$. Then $x\not\le y$.
\begin{itemize}
\item If $x,y\in P$, then $u_x\not\subword u_y$ by choice of $u_x$ and $u_y$. In particular, we have 
$v_x=b^{f(x)\cdot (m+1)} u_x \not\subword b^{f(y)\cdot (m+1)} u_y=v_y$.
\item If $x\in P$ and $y\in L$, then $v_y\in b^*$ and $a\subword v_x$. Thus $v_x\not\subword v_y$.
\item If $x\in L$ and $y\in P$, say $x=\ell_i$. Then $x\not\le y$ means that $\ell_i\notin\down\{y\}\cap L$ and hence $f(y)<i$.
Note that $v_y=b^{f(y)(m+1)} u_y$ and that $|u_y|_b\le m$. Therefore $|v_y|_b\le f(y)(m+1)+m<i(m+1)$ and hence 
$v_x=v_{\ell_i}=b^{i(m+1)}\not\subword v_y$.
\item If $x,y\in L$, then $x=y_i$ and $y=y_j$ with $j<i$. Then clearly $v_x\not\subword v_y$.
\end{itemize}
\item Constraints $x\in b^*$ or $x\in b^+$ with $x\in L$ are of course satisfied.
\item Constraints $x\not\subword a$ with $x\in P$ are satisfied because $|u_x|_a\ge m\ge 2$ for every $x\in P$.
\end{itemize}
This established our claim and thus the \lcnamecref{decidable-short}.
\end{proof}

This raises an interesting question: For which
sets $\{W_1,W_2,\ldots\}\subseteq A^*$ of constants is the truth
problem for $\Sigma_1$ sentences over $\FOconst[W_1,W_2,\ldots]{A}$
decidable?

\subsection{The $\frag{2}{1}$ fragment}

Our next result is that if we allow one more quantifier alternation, then already
one variable without alternation bound is sufficient to prove undecidability.

\begin{theorem}
Let $|A|\ge 2$ and $a\in A$. For each recursively enumerable set
$S\subseteq\N$, there is a $\frag{2}{1}$ formula $\varphi$ over the structure
$\FOconst{A}$ with one free variable such that $\assignments{\varphi}=\{a^k
\mid k\in S\}$. In particular, the truth problem for $\frag{2}{1}$ is undecidable.
\end{theorem}
\begin{proof}
\begin{enumerate}
\item We can express ``$|u|_a\le |v|_a$'' in $\fragpif{1}{0}$:
\[ \forall x\in a^* \colon x\not \subword u \vee x\subword v. \]
Hence, we can express ``$|u|_a=|v|_a$'' in $\fragpif{1}{0}$.
\item We can express ``$|u|_a>|v|_a$'' in $\fragpif{1}{0}$: This follows from the fact that
$|u|_a\le|v|_a$ is expressible in $\fragf{1}{0}$.
\item We can express ``$|u|_a\ne |v|_a$'' in $\fragpif{1}{0}$ according to the previous item.
\item We can express ``$u\in a^* \wedge v\in (bu)^*$'' in $\fragpi{1}{1}$. It clearly
suffices to express ``$u\in a^* \wedge v\in (bu)^+$'' in $\fragpi{1}{1}$. Consider the formula
\begin{multline*}
v\in b\{a,b\}^* 
~\wedge~ \forall x\in b^+a^*b^*\colon
\bigl[|x|_b\ne |v|_b \\
\vee~ (|x|_a>|u|_a\wedge x\not\subword v) 
~\vee~ (|x|_a\le |u|_a\wedge x\subword v) \bigr].
\end{multline*}
Note that ``$v\in b\{a,b\}^*$'' is expressible in $\fragpif{1}{0}$ because
$v\in a\{a,b\}^*$ is expressible in $\fragf{1}{0}$ (see \cref{undec2:last} in
the proof of \cref{subword:undecidable}).

Moreover, notice that since $b^*a^*b^*=\{a,b\}^*\setminus \LL{a b a}$, the
language $b^+a^*b^* = (b^*a^*b^*)\cap (\LL{b a}\cup (\LL{b}\setminus \LL{a}))$
is piecewise testable and thus definable in $\fragf{0}{0}$.
\item We can express ``$|u|_a=|v|_b$'' in $\fragf{2}{1}$:
\[\exists x\in (a b)^* \colon |u|_a=|x|_a \wedge |v|_b=|x|_b. \]
\item We can express ``$\exists m,n\colon u=a^m \wedge v=a^n \wedge w=a^{m\cdot n}$'' in $\frag{2}{1}$:
\begin{align*}
u\in a^*\wedge v\in a^*\wedge w\in a^*
 \wedge \exists y\in b^*\colon\exists x\in (bu)^*\colon \\
|x|_b=|y|_b \wedge |y|_b=|v|_a\wedge |x|_a=|w|_a.
\end{align*}
Note that we employ the variable $y$ because directly expressing $|x|_b=|v|_a$
(using the previous item) would require an additional alternation unbounded
variable besides $x$, but we can only use one.
\item Recall that ``$|w|_a=|u|_a+|v|_a$'' is expressible in $\fragf{1}{0}$ (see
\cref{undec2:add} of \cref{subword:undecidable}) and hence ``$\exists m,n\colon
u=a^m \wedge v=a^n\wedge w=a^{m+n}$'' in $\frag{1}{0}$.  Thus, we can implement
Diophantine equations as in the proof of \cref{subword:undecidable}.
\end{enumerate}
\end{proof}

\subsection{The $\Sigma_2$ fragment over two letters in the pure logic}
The final result in this section settles the question of how many letters are
needed to make the $\Sigma_2$ fragment of $\FOpure{A}$ undecidable. We show
here that two letters suffice. Observe that if $|A|=1$, $\FOpure{A}$ can be
interpreted in $\FO(\N, <)$ and is thus decidable.

Let $\mu\colon A^*\to A^*$ be a map. It is called a \emph{morphism} if
$\mu(uv)=\mu(u)\mu(v)$ for all $u,v\in A^*$.  It is an
\emph{anti-morphism} if $\mu(uv)=\mu(v)\mu(u)$ for all $u,v\in A^*$.
Finally, it is an \emph{automorphism} of $(A^*,\subword)$ if for any
$u,v\in A^*$, we have $u\subword v$ if and only if
$\mu(u)\subword\mu(v)$. 

Note that if we have no constants, we cannot define a language in
$a^*$, because all definable subsets are closed under automorphisms of
$(A^*,\mathord{\subword})$. It will be useful for the next proof to
have a classification of all automorphisms of $(A^*,\subword)$.  The
following is shown implicitly by Kudinov et. al.~\cite{kudinov2010},
but we include a short proof for completeness.
\begin{lemma}\label{automorphisms}
  A map $\mu$ is an automorphism of $(A^*,\subword)$ if and only if
  \begin{enumerate}[label=(\roman*)]
  \item $\mu$ is either a morphism or an
    anti-morphism and
  \item $\mu$ permutes $A$.
  \end{enumerate}
\end{lemma}
\begin{proof}
  Clearly, maps as described in the \lcnamecref{automorphisms} are
  automorphisms.  Assume $\mu$ is an automorphism.  Since $\mu$ has to
  preserve the minimal element, we have
  $\mu(\varepsilon)=\varepsilon$.  It also has to preserve the set of
  minimal elements of $A^*\setminus \{\varepsilon\}$, hence the set
  $A$. Repeating this argument yields that $\mu$ has to preserve
  length. Therefore, according to \cref{shorter-subwords}, if $\mu$ is
  identical on $A^{\le n}$ for $n\ge 2$, then it is the identity on
  $A^{\le n+1}$. By induction, this implies that if an automorphism is
  identical on $A^{\le 2}$, then it is the identity on $A^*$. Hence,
  if two automorphisms agree on $A^{\le 2}$, then they are the
  same. It therefore suffices to show that every automorphism $\mu$
  agrees on $A^{\le 2}$ with a map as described in the
  \lcnamecref{automorphisms}.

  Since $\mu$ preserves the set $A$, the map $\pi=\mu|_A$ is a
  permutation of $A$.  Moreover, for any $a,b\in A$, we have
  $\mu(ab)=\mu(a)\mu(b)$ or $\mu(ab)=\mu(b)\mu(a)$.
  If $\mu(ab)=\mu(a)\mu(b)$, then we cannot have
  $\mu(bc)=\mu(c)\mu(b)$, because the two words $ab$ and $bc$ have
  only one common upper bound of length three (namely $abc$), whereas
  the words $\mu(a)\mu(b)$ and $\mu(c)\mu(b)$ have two, namely
  $\mu(c)\mu(a)\mu(b)$ and $\mu(a)\mu(c)\mu(b)$. Therefore, if
  $\mu(ab)=\mu(a)\mu(b)$, then $\mu(bc)=\mu(b)\mu(c)$. In particular,
  if $\mu(ab)=\mu(a)\mu(b)$, then we have $\mu(cd)=\mu(c)\mu(d)$ for
  all $c,d\in A$. Hence on $A^{\le 2}$, $\mu$ agrees with a map as
  decribed.
\end{proof}

\begin{corollary}\label{undecidable:sigma2}
Let $|A|\ge 2$.  For each recursively enumerable set $S\subseteq\N$, there is a
$\Sigma_2$ formula $\tau$ over the structure $\FOpure{A}$  that defines the language 
\[ \assignments{\tau}=\{a^k\mid a\in A,~k\in S\}. \]
 In particular, the truth
problem for $\Sigma_2$ is undecidable.
\end{corollary}
\begin{proof}
Fix a letter $a\in A$.  Let $S\subseteq \N$ be recursively enumerable, let
$\varphi$ be the $\frag{1}{3}$ formula provided by \cref{subword:undecidable}
with one free variable $x$ and with $\assignments{\varphi}=\{a^k \mid k\in
S\}$, and let $w_1,\ldots,w_m\in A^*$ be the constants used in the formula
$\varphi$. 

It was shown in~\cite{KS-fosubw} that from $w_1,\ldots,w_m$, one can construct
a $\Sigma_2$ formula $\psi$ over $\FOpure{A}$ with free variables
$V=\{x_1,\ldots,x_m\}$ such that for $\alpha\in (A^*)^V$, we have
$\alpha\in\assignments{\psi}$ if and only if there is an automorphism $\overline{\cdot}\colon
A^*\to A^*$  such that $\alpha(x_i)=\overline{w_i}$ for every $i\in[1,m]$.

Let $\varphi'$ be the formula obtained from $\varphi$ by replacing
every occurrence of $w_i$ with $x_i$. Moreover, let $\tau\equiv\exists
z_1,\ldots,z_r\colon \psi\wedge \varphi'$.

Then, $\tau$ is clearly a $\Sigma_2$ formula and has exactly one free variable, say $x$.
We claim that 
\begin{equation}\assignments{\tau}=\{b^k \mid b\in A, ~k\in S\}.\label{undecidability:sigma2:assign} \end{equation}
If $k\in S$, then $a^k\in\assignments{\varphi}$ and hence clearly
$b^k\in\assignments{\tau}$ for each $b\in A$.  Moreover, if
$w\in\assignments{\tau}$, then for some $\alpha\in\assignments{\psi}$,
we have $\alpha_w\models\varphi'$, where $\alpha_w$ denotes the
assignment with $\alpha_w|_V=\alpha$ and $\alpha(x)=w$. This means,
there is an automorphism $\overline{\cdot}$ of $(A^*,\subword)$ such
that $\alpha(x_i)=\overline{w_i}$ for $i\in[1,m]$. Therefore, there is
some $w'\in A^*$ that satisfies $\varphi$ such that
$w=\overline{w'}$. In particular, $w'=a^k$ for some $k\in S$ and hence
$w=b^k$ for some $b\in A$. This proves \cref{undecidability:sigma2:assign}. 

We can now proceed as
in \cref{undecidable:single} to show undecidability of the truth
problem.
\end{proof}

\section{Complexity}\label{complexity}
In this section, we study the complexity of the truth problem for the
$\frag{i}{j}$ fragments of $\FOconst{A}$. 

\subsection{Complexity of $\frag{i}{0}$}
We begin with the case $j=0$.  In the
following, $\WeakExp{n}$ denotes the $n$-th level of the weak $\EXP$
hierarchy~\cite{hemachandra89,gottlob95b}.
\begin{theorem}
\label{thm:presburger-eq}
If $|A|\ge 2$, then the truth problem for $\frag{i}{0}$ is
$\NP$-complete for $i=1$ and $\WeakExp{i-1}$-complete for $i>1$.
\end{theorem}
We provide a polynomial inter-reduction with the $\Sigma_i$
fragment of 
$\FO(\N,0,1,\mathord{+},\mathord{<})$, a.k.a.
Presburger Arithmetic (PA), for which Haase~\cite{Haase2014} has
recently proven $\WeakExp{i-1}$-completeness for $i>1$. The $\Sigma_1$ fragment of PA is $\NP$-complete~\cite{Oppen1978}.

The reduction from PA to
$\frag{1}{0}$ fixes a letter $a\in A$ and encodes every number $k\in\N$ by
$a^k$. Addition can then be expressed in $\frag{1}{0}$ (\cref{undec2:add} of
\cref{subword:undecidable}). Note that although this ostensibly works with one
letter, we need another letter in $A$ to express addition.  This is crucial: If
$|A|=1$, then $\FOconst{A}$ is just $\FO(\N,<)$, which has a $\PSPACE$-complete
truth problem~\cite{FerranteRackoff1979,Stockmeyer1974}. Moreover, an
inspection of the proof of \cref{subword:undecidable} shows that an alternation
bound of $\ell=2$ suffices to define addition, which is tight: if we only use a
bound of $1$, we can also easily reduce to $\FO(\N,<)$.

The reduction from $\frag{i}{0}$ to Presburger arithmetic encodes a word $w$
known to belong to $(a_1^*\cdots a_n^*)^\ell$, i.e., of the form
\[ \prod_{i=1}^\ell\prod_{j=1}^n a_j^{x_{i,j}}, \]
by the vector
$(x_{1,1},\ldots)\in\N^{\ell\cdot n}$ of exponents.
With this encoding, it suffices to show how to
express literals $w\subword w'$ (and also $w\not\subword w'$) by polynomial-size existential
Presburger formulas for $w,w'\in (a_1^*\cdots a_n^*)^\ell$.  For a vector
$x=(x_{1,1},\ldots,x_{\ell,n})$ from $\N^{\ell\cdot n}$, let
$ w_x =  \prod_{i=1}^\ell \prod_{j=1}^n a_j^{x_{i,j}}$. 
\begin{proposition}\label{subwords-in-presburger}
There are existential Presburger formulas $\varphi$ and $\psi$ of size polynomial in $n$ and $\ell$ such that
\begin{align*}
\varphi (x_{1,1},\ldots,x_{\ell,n},y_{1,1},\ldots,y_{\ell,n}) \iff w_x\subword w_y, \\
\psi    (x_{1,1},\ldots,x_{\ell,n},y_{1,1},\ldots,y_{\ell,n})  \iff w_x\not\subword w_y.
\end{align*}
\end{proposition}

Let us briefly describe these formulas. Let $I=[1,\ell]\times[1,n]$ and order
 the pairs $(i,j)\in I$ 
lexicographically:  $(i',j')<(i,j)$ if $i'<i$
or $i=i'$ and $j'<j$. This captures  the order of
the $a_j^{x_{i,j}}$  factors in $w_x$. We now define
formulas $\tau_{i,j}$ and $\eta_{i,j}$ where 
the $t_{i,j,k}$'s and $e_{i,j,k}$'s  are  extra free variables:
\begin{align*}
  \tau_{i,j} \colon && \bigwedge_{1 \le k \le \ell} t_{i,j,k}
  &= \begin{cases}
    0  \quad \text{if $e_{i',j',k'}>0$ for some} \\
       \quad \quad \quad (i',j')<(i,j) \text{ and } k'>k \\
    y_{k,j}  - \sum_{i' = 1}^{i-1} e_{i',j,k} \quad \text{otherwise} \end{cases} \\
  \eta_{i,j} \colon && \bigwedge_{1 \le k \le \ell} e_{i,j,k}
  &= \min\left\{ t_{i,j,k}~,~ x_{i,j}-\sum_{r=1}^{k-1} e_{i,j,r} \right\}
\end{align*}
These expressions define the leftmost embedding of $w_x$ into
$w_y$: the variable $t_{i,j,k}$ describes how many letters from $a_j^{y_{k,j}}$
are available for embedding the $a_j^{x_{i,j}}$ factor of $w_x$ into $w_y$. The variable $e_{i,j,k}$ counts how
many of these available letters are actually used for the $a_j^{x_{i,j}}$ factor
in
the left-most embedding of $w_x$ into $w_y$.  Since $i$ and $j,k$ are bounded by $n$
and $\ell$, we have polynomially many formulas of polynomial size. 

Define $\xi=\bigwedge_{(i,j)\in I} \tau_{i,j}\wedge \eta_{i,j}$ and the formulas
$\varphi,\psi$ as:
\begin{gather}
\tag{$\varphi$}
\begin{array}{c}
  \exists t_{1,1,1}\cdots\exists t_{\ell,n,\ell}\quad
        \\
  \quad\exists e_{1,1,1}\cdots \exists e_{\ell,n,\ell}
\end{array}
\colon
 \xi \wedge \bigwedge_{(i,j)\in I} \left(x_{i,j}\le \sum_{k=1}^\ell e_{i,j,k}\right)
\\
\tag{$\psi$}
\begin{array}{c}
  \exists t_{1,1,1}\cdots\exists t_{\ell,n,\ell}\quad
        \\
  \quad\exists e_{1,1,1}\cdots \exists e_{\ell,n,\ell}
\end{array}
\colon
 \xi \wedge \bigvee_{(i,j)\in I} \left(x_{i,j}>\sum_{k=1}^\ell e_{i,j,k}\right)
\end{gather}
Since formulas $\tau_{i,j}$ and $\eta_{i,j}$ are inductive
equations that uniquely define the values of $t_{i,j,k}$ and $e_{i,j,k}$
as functions of the $x$ and $y$ vectors, $\psi$ is equivalent to the negation
of $\varphi$.  Moreover, $\varphi$ expresses that there is enough room to embed
each factor $a_j^{x_{i,j}}$ in $w_y$, i.e., that $w_x\subword w_y$ as claimed,
and both formulas are easily constructed in polynomial time.

\subsection{Complexity of $\frag{1}{1}$}

\begin{theorem}
  \label{thm:sat-1-1}
The truth problem for the $\frag{1}{1}$ fragment is $\NP$-complete.
\end{theorem}
Of course, hardness is inherited from $\frag{1}{0}$.  Conversely,
$\NP$-membership is shown by a reduction to the $\frag{1}{0}$ fragment. For
this reduction, we first explain how a single ``unbounded'' word can be made
alternation bounded while respecting its relationships with  other
alternation bounded words. 

For this we use a slightly different measure of alternation levels for words: we
factor words in \emph{blocks} of repeating letters, writing
$u=\prod_{i=1}^k{a_i}^{\ell_i}$ with $\ell_i>0$ and $a_i\not= a_{i+1}$ for all
$i$. By ``an $a$-block of $u$'' we mean an occurrence of a factor $a_i^{\ell_i}$ with $a_i=a$.
We note that requiring some bound in the number of blocks is equivalent to
bounding the number of alternations when it comes to defining the $\frag{i}{j}$
fragments. However, counting blocks is more precise.

\begin{lemma}
\label{lem-k-blocks}
  Let $t, x_1, \dots, x_n, y_1, \dots, y_m \in A^*$ such that:
  \begin{itemize}
  \item for all $i$, $x_i \subword t$,
  \item for all $j$, $y_j \not\subword t$,
  \item for all $i$ and $j$, $x_i$ and $y_j$ have less than $\ell$ blocks,
  \item $t$ has $k > (m+n)\cdot \ell + |A|$ blocks.
  \end{itemize}
  Then there exists $t' \in A^*$ such that:
  \begin{itemize}
  \item for all $i$, $x_i \subword t'$,
  \item for all $j$, $y_j \not\subword t'$,
  \item $t'$ has either $k-1$ or $k-2$ blocks.
  \end{itemize}
\end{lemma}

\begin{proof}
  Given $u \in A^*$, we write $\Ima u$ for the image of the
  left-most embedding of $u$ into $t$. This is a set of positions in $t$ and, in case
  $u\not\subword t$, these positions only account for the longest prefix of $u$
  that can be embedded in $t$. In particular, and since we assumed
$x_i\subword t$ and $y_j\not\subword t$, then
$|\Ima x_i|=|x_i|$ and $|\Ima y_j|<|y_j|$ for all $i,j$.

Let $b_0$ be an $a$-block of $t$. This block is said to be \emph{irreducible} if
either (1) it is the last, i.e.\ right-most, $a$-block of $t$, or (2) writing
$t$ under the form $t = t_0 b_0 t_1 b_1 t_2$ where $b_1$ is the next $a$-block,
i.e.\ $a \notin t_1$, one of the following holds:
    \begin{itemize}
    \item there is  some $i$ s.t.\ $b_0 \cap \Ima x_i\neq\emptyset$ and $t_1 \cap \Ima x_i\neq\emptyset$.
    \item there is $j$ s.t.\ $b_0 \cap \Ima y_j = \emptyset$ and $t_1 \cap \Ima y_j \neq \emptyset$ and $b_1 \cap \Ima y_j \neq \emptyset$.
    \end{itemize}
Otherwise $b_0$ is said to be \emph{reducible}.
\\

\noindent
\textit{Claim:} $t$ contains a reducible block.

Indeed, every irreducible block is either a right-most $a$-block for some $a$,
or can be associated with a letter alternation in some $x_i$, or in some
$y_j$. Furthermore, this association is injective. Thus there are at most
$(n+m)\cdot \ell$ irreducible blocks that are not right-most (and at most $|A|$
right-most blocks). Since $k > (n+m)\cdot \ell + |A|$, $t$ has a reducible
block.

So let us pick one such reducible block, say an $a$-block $b_0$, 
write $t$ under the form $t = t_0 b_0 t_1 b_1 t_2$ as above, and
let $t' = t_0 t_1 b_0 b_1 t_2$. \\

\noindent
\textit{Claim:} $t'$ fulfills the requirements of
\Cref{lem-k-blocks}.

Since $b_1$ is an $a$-block, $b_0b_1$ is now a block of $t'$ and $t'$ has less
than $k$ blocks. Moreover, the only other possible block merge is in $t_0t_1$,
thus $t'$ has at least $k-2$ blocks. We now show that  $x_i\subword t'$
and $y_j\not\subword t'$ for all $i,j$.
  \begin{itemize}
  \item Pick some $i$. Since $x_i \subword t$, there
    is a unique decomposition $x_i = u_0 u_1 u_2 u_3 u_4$ of $x_i$ such that
    $\Ima u_0 \subseteq t_0$, $\Ima u_1 \subseteq b_0$, $\Ima u_2 \subseteq
    t_1$, $\Ima u_3 \subseteq b_1$ and $\Ima u_4 \subseteq t_2$. 
Since $b_0$ is reducible one of $\Ima u_1$ or $\Ima u_2$ is empty.
Thus one of $u_1$ or $u_2$ is the empty word, allowing $x_i\subword t'$.

  \item Assume, by way of contradiction,  that for some $j$, $y_j \subword t'$.
Let $z_1$ be the maximal prefix of $y_j$ that embeds into $t$.
We proceed to show that $b_0$ is irreducible.
\begin{itemize}
\item First, $b_0 \cap \Ima z_1 = \emptyset$. Otherwise, since $a \notin t_1$,
  the left-most embedding of $z_1$ into $t' = t_0 t_1 b_0 b_1 t_2$ does not use
  $t_1$ at all and we would have $y_j \subword t_0 b_0 b_1 t_2 \subword t$.
\item Secondly, $t_1 \cap \Ima z_1$ is not empty. If it were, since $b_0$ is
  made of $a$'s only and $a \notin t_1$, the left-most embedding of $z_1$ into
  $t_0 t_1 b_0 b_1 t_2$ would not use $t_1$ and  again we would have $y_j \subword t_0 b_0 b_1 t_2 \subword t$.
\item Lastly, $b_1 \cap \Ima z_1 \neq \emptyset$. 
Otherwise, the already established fact $b_0\cap \Ima z_1=\emptyset$
implies that $y_j$ embeds not only in $t'$ but in $t_0t_1t_2$, which is a subword of $t$.
\end{itemize}
  Since  $b_0$ is reducible, we conclude  that the original assumption that
  $y_j\subword t'$ does not hold, i.e., that $y_j\not\subword t'$ as required.
\end{itemize}

\end{proof}

We now proceed to prove \Cref{thm:sat-1-1}.  Let $\varphi$ be a
$\frag{1}{1}$ sentence, where $t$ is the only variable which is not
alternation bounded.  As a first step of our $\NP$ algorithm, we guess
the set of literals occurring in $\varphi$ that is satisfied. After
guessing this subset, we check whether the formula would be satisfied
if exactly this subset were true (which essentially amounts to
evaluating a formula in propositional logic). If this is the
case, it remains to check whether it is possible to choose words for
all existential quantifiers of $\varphi$ so that exactly this subset
of literals is true. This means, we are left with the task of checking
satisfiability of a formula $\varphi$ of the form $\varphi\equiv\exists
t\colon \psi$: Here, $\psi$ begins with existential quantifiers for
alternation bounded variables, which are followed by a conjunction of
literals.

Every literal in $\psi$ that involves $t$ is of one of the following types:
  \begin{typelist}
  \item\label{type:bt} $x \subword t$, where $x$ is an alternation bounded variable,
  \item\label{type:bnt} $y \not\subword t$, where $x$ is an alternation bounded variable,
  \item\label{type:tnb} $t \not\subword u$, where $u$ is an alternation bounded variable,
  \item\label{type:tb} $t \subword z$, where $z$ is an alternation bounded variable,
  \item\label{type:tt} $t \subword t$,
  \item\label{type:tnt} $t \not\subword t$.
  \end{typelist}
  Assertions of \cref{type:tt,type:tnt} can be replaced by their truth
  value. If a literal $t\subword z$ of \cref{type:tb} occurs in
  $\psi$, then $\varphi$ is equivalent to $\exists t \in (a_1^* \cdots
  a_n^*)^\ell\colon \psi$,
  where $\ell$ is the alternation bound of variable $z$.

  We can thus assume that only literals of \cref{type:bt,type:bnt,type:tnb} occur in $\psi$. Let $n$
  be the number of variables $x$ that occur in literals of \cref{type:bt}, $m$ the
  number of variables $y$ that occur in literals of \cref{type:bnt}, $\ell$ the
  maximum alternation level of these variables, and $k$ the maximum alternation
  bound of all variables $u$ that appear in literals of \cref{type:tnb}.  Let 
  \[ p = \max\{ (m+n)\cdot (\ell \cdot |A|) + |A|, ~k\cdot |A| +3\}. \]
  (here $\ell$ and $k$
  are multiplied by $|A|$ to obtain a number of blocks from a maximum
  alternation). Then $\varphi$ is equivalent to $\exists t \in (a_1^* \cdots
  a_n^*)^p\colon \psi$. Indeed, if the restricted formula has a solution,
  it is a solution for $\psi$. Conversely, if $\psi$ is satisfiable via some
  $t\in A^*$ having more than $p$ blocks, then by \Cref{lem-k-blocks}, one can
  also use a $t$ having between $k\cdot |A|$ and $p$ blocks. The fact that $t$
  has more than $k\cdot |A|$ blocks ensures that all literals $t \not\subword u$
  are still satisfied.

  Finally, we can replace every $\exists t$ in $\varphi$ by a bounded
  quantification and obtain an equivalent formula in $\frag{1}{0}$ which proves
  \Cref{thm:sat-1-1}
\\

To the authors' knowledge, the following was not known.
  \begin{corollary}
\label{coro-interesting}
    If PT languages are represented as boolean combinations of sets of the form
$\LL{w}$ with $w\in A^*$, then their non-emptiness problem is $\NP$-complete.
  \end{corollary}
Membership in $\NP$ follows from \cref{thm:sat-1-1}. For hardness, we can
reduce CNF-SAT as follows. We encode an assignment
$\alpha\colon \{x_1,\ldots,x_n\}\to\{0,1\}$ as a word
$b a^{\alpha(x_1)}b a^{\alpha(x_2)}\cdots b a^{\alpha(x_n)}$. With literals
$x_i$ and $\neg x_i$, we associate the languages 
$K_{x_i}=\LL{b^{i}a b^{n-i}}$ and $K_{\neg x_i}=\{a,b\}^*\setminus
\LL{b^{i}a b^{n-i}}$.  A clause $C=L_1\vee \cdots \vee L_m$ is then translated to
$K_C=K_{L_1}\cup\cdots\cup K_{L_n}$ and a conjunction of clauses  $C_1\wedge\cdots\wedge C_k$
is satisfiable if, and only if, the PT language $(b(a+\eword))^n\cap K_{C_1}\cap \cdots \cap
K_{C_k}$ is nonempty.

In particular, given a finite
number of PT languages, the problem of deciding whether they intersect non-vacuously is
$\NP$-complete. This is in contrast with general regular languages represented
by DFAs (or NFAs), for which the problem is well-known to be $\PSPACE$-complete~\cite{Kozen1977}.

\subsection{Complexity of $\frag{1}{2}$}

Our next result is an upper bound for the truth problem of $\frag{1}{2}$.
\begin{theorem}\label{complexity:two-variables:counterrelation}
The truth problem for $\frag{1}{2}$ is in $\NEXP$.
\end{theorem}
We prove \cref{complexity:two-variables:counterrelation} in two steps.  The
first step of our decidability result is to transform a $\frag{1}{2}$ formula
into a system of constraints where the relations among those variables without
an alternation bound have a tree shape. In the second step, we exploit the tree
shape to construct an exponential-size counter automaton for the set of
satisfying assignments.

\subsection*{Tree-shaped constraints} 
Let $A=\{a_1,\ldots,a_n\}$ and let $V$ be a set of variables. A
\emph{constraint system} is a set of constraints of the form $x\subword y$,
$x\not\subword y$, $x=y$, $x\in(a_1^*\cdots a_n^*)^\ell$, or $x=w$, where
$x,y\in V$, $\ell\in\N$ and $w\in A^*$. A constraint of the form $x\subword y$,
$x\not\subword y$, or $x=y$ is also called \emph{$(x,y)$-constraint} or
\emph{$(y,x)$-constraint}.  Constraints of the form $x\in (a_1^*\cdots
a_n^*)^\ell$ are called \emph{alternation constraints}.  The set of assignments
$\alpha\in (A^*)^V$ that satisfy $S$ is denoted by $\assignments{S}$.  For a
subset $U\subseteq V$, by existentially quantifying all variables outside of
$U$, the constraint system $S$ also defines a set of assignments in $(A^*)^U$, which we denote by
$\assignments[U]{S}$.

For a constraint system $S$ over $V$, we define the graph $\Gamma(S)=(V,E)$
where $\{x,y\}\in E$ if and only if $S$ contains an $(x,y)$-constraint.  We say
that $S$ is \emph{tree-shaped} if $\Gamma(S)$
is a forest. Furthermore, $S$ is called \emph{alternation bounded} if every
variable occurring in $S$ also has an alternation constraint in $S$.

\begin{proposition}\label{complexity:constraint-decomposition}
For any disjunction-free $\frag{1}{2}$-formula $\varphi$, one can construct
polynomial-size constraint systems $T$ and $S$ over variables $V'\supseteq \freevar{\varphi}$ such that 
\begin{enumerate}
\item $T$ is tree-shaped,
\item $S$ is alternation bounded, and
\item $\assignments[\freevar{\varphi}]{T\cup S}=\assignments{\varphi}$.
\end{enumerate}
\end{proposition}
\begin{proof}
We will need the notion of quotients of constraint systems. The idea is to
identify certain pairs of variables.  Suppose $S$ is a constraint system over
$V$. Furthermore, let $\sim\subseteq V\times V$ be an equivalence relation that
specifies which variables we want to identify with each other. Then we define
the \emph{quotient} $S/\mathord{\sim}$ as a constraint system over the variable
set $V/\mathord{\sim}$ with the constraints
\begin{align*}
S/\mathord{\sim} ~= &~\{ [x]\delta [y] \mid \delta\in\{\subword,\not\subword\}, ~(x\delta y)\in S\}. \\
                   &\cup~ \{[x] = w \mid  (x=w)\in S \} \\
                   &\cup~ \{[x]\in (a_1^*\cdots a_n^*)^\ell \mid (x\in (a_1^*\cdots a_n)^\ell) \in S\}
\end{align*}

In the course of constructing the constraint systems, it will be convenient to
assume that two constraint systems are defined over disjoint sets of variables.
To this end, we need some notation to state that a constraint system is
equivalent to a formula even though its variables have different names.
Suppose we have a constraint system $S$ over the set of variables $V'$ and
$\psi\colon\freevar{\varphi}\to V'$ is an injective map. Then, via $\psi$, the
formula $\varphi$ defines a set of assignments
$\assignments[\psi]{\varphi}\subseteq (A^*)^{\im{\psi}}$.  We say that
$\varphi$ and $S$ are \emph{$\psi$-equivalent} if
$\assignments[\im{\psi}]{S}=\assignments[\psi]{\varphi}$.

We may clearly assume that all literals are of the form
$x\subword y$, $x\not\subword y$, or $x=w$ for $w\in A^*$ (and there are no
literals $w\subword x$ etc.).

We show the following stronger statement. Let $B$ be the set of variables in
$\varphi$ that are alternation bounded. For each disjunction-free
$\frag{1}{2}$-formula $\varphi$, there is a set of variables $V'$, constraint
systems $T$ and $S$ over $V'$, an injective map $\psi\colon
\freevar{\varphi}\to V'$ such that the following holds. If $B'\subseteq V'$
denotes the set of variables for which there is an alternation bound in $S$,
then
\begin{enumerate}[label=(\roman*)]
\item\label{treeshape:tree} $T$ is tree-shaped,
\item\label{treeshape:alt} $S$ is alternation-bounded,
\item\label{treeshape:equi} $T\cup S$ and $\varphi$ are $\psi$-equivalent,
\item\label{treeshape:map} for every $x\in B$ we have $\psi(x)\in B'$, and
\item\label{treeshape:two} if $|\freevar{\varphi}\setminus B|=2$ with $\freevar{\varphi}\setminus B=\{x,y\}$, then
$\psi(x)$ and $\psi(y)$ are either neighbors in $\Gamma(T)$ or in distinct components.
\end{enumerate}

To prove this statement by induction, we need to consider three cases.
\begin{enumerate}
\item Literals, i.e. $x\subword y$, $x\not\subword y$, or $x=w$. There are only two
variables so we can just take the literal as the set $T$ and let $S$
contain the global alternation constraints for the variables in the literal.
\item Existentially quantified formulas $\exists x\colon \varphi$. Here, we
just reduce the set of free variables, so it suffices to adjust the map $\psi$.
\item Conjunctions $\varphi\equiv\varphi_0\wedge \varphi_1$. Suppose we have
constructed $T_i,S_i,V'_i,\psi_i,B'_i$ as above for $i=0,1$. We may clearly
assume $V_0\cap V_1=\emptyset$.  We construct $T,S,V',\psi,B'$ as follows.  Let
$\sim\subseteq (V'_0\cup V'_1)\times (V'_0\cup V'_1)$ be the smallest
equivalence relation with $\psi_0(x)\sim \psi_1(x)$ for all
$x\in \freevar{\varphi}\setminus B$. Then we take $V'=(V'_0\cup V'_1)/\mathord{\sim}$
and define $T=(T_0\cup T_1)/\mathord{\sim}$. Moreover, let
\[ S=S_0\cup S_1\cup \{[\psi_0(x)]=[\psi_1(x)] \mid x\in\freevar{\varphi}\cap B\}. \]
Moreover, we choose $\psi\colon\freevar{\varphi}\to V'$ so  that
$\psi(x)=\psi_0(x)$ if $x\in\freevar{\varphi_0}$ and $\psi(x)=\psi_1(x)$ if
$x\notin\freevar{\varphi_0}$.

It is clear that \cref{treeshape:alt,treeshape:equi,treeshape:map} above are
satisfied.  It remains to verify \cref{treeshape:tree,treeshape:two}.
We distinguish the following cases.
\begin{itemize}
\item $|\freevar{\varphi_i}\setminus B|\le 1$ for some $i\in\{0,1\}$. Then
there is at most one variable in $V'_i$ that is identified with a variable in
$V'_{1-i}$ by $\sim$. Hence, $\Gamma(T)$ is obtained from $\Gamma(T_0)$ and
$\Gamma(T_1)$ either by disjoint union or by identifying one vertex from
$\Gamma(T_0)$ with one vertex from $\Gamma(T_1)$. In any case, $\Gamma(T)$ is a
forest. Hence, \cref{treeshape:tree} is satisfied.

Let us now show \cref{treeshape:two}. Hence, assume $|\freevar{\varphi}\setminus
B|=2$ with $\freevar{\varphi}\setminus B=\{x,y\}$.

Clearly, if $\freevar{\varphi_0}\setminus B$ and $\freevar{\varphi_1}\setminus
B$ are disjoint, then no variables are identified and hence $\psi(x)$ and
$\psi(y)$ are in distinct components of $\Gamma(T)$. Hence, we assume that
$\freevar{\varphi_0}\setminus B$ and $\freevar{\varphi_1}\setminus B$ have a
variable in common, say $x$.

Since $|\freevar{\varphi_0}\setminus B|\le 1$, this means
$\freevar{\varphi_1}\setminus B=\{x,y\}$ and hence $\psi_1(x)$ and $\psi_1(y)$
are neighbors in $\Gamma(T_1)$ or they are in distinct components of $\Gamma(T_1)$.

$\Gamma(T)$ is obtained from $\Gamma(T_0)$ and $\Gamma(T_1)$ by identifying
$\psi_0(x)$ and $\psi_1(x)$.  Therefore, $\psi(x)$ and $\psi(y)$ are neighbors
in $\Gamma(T)$ if and only if they are neighbors in $\Gamma(T_1)$.  Moreover,
they are in disjoint components of $\Gamma(T)$ if and only if they are in
disjoint components of $\Gamma(T_1)$. This proves \cref{treeshape:two}.
\item $|\freevar{\varphi_0}|=|\freevar{\varphi_1}|=2$. Write
$\freevar{\varphi_0}\setminus B=\freevar{\varphi_1}=\{x,y\}$. Note that
$\Gamma(T)$ is obtained from $\Gamma(T_0)$ and $\Gamma(T_1)$ by identifying
$\psi_0(x)$ with $\psi_1(x)$ and identifying $\psi_0(y)$ with $\psi_1(y)$.
Since for each $i\in\{0,1\}$ we know that $\psi_i(x)$ and $\psi_i(y)$ are
either neighbors in $\Gamma(T_i)$ or in distinct components, this clearly
implies that $\Gamma(T)$ is a forest. Hence, we have shown \cref{treeshape:tree}.

Moreover, if for some $i\in\{0,1\}$, $\psi_i(x)$ and $\psi_i(y)$ are neighbors
in $\Gamma(T_i)$, then $\psi(x)$ and $\psi(y)$ are neighbors in $\Gamma(T)$.
Otherwise, $\psi(x)$ and $\psi(y)$ are in disjoint components of $\Gamma(T)$.
This proves \cref{treeshape:two}.
\end{itemize}
\end{enumerate}
\end{proof}

\subsection*{Counter automata}
In the next step, we exploit the decomposition into a tree-shaped constraint
system and an alternation-bounded constraint system to reduce satisfiability to
non-emptiness of counter automata.
To this end, we use a type of counter automata known as \emph{Parikh
automata}~\cite{KlaedtkeRuess2003,CaFiMcK2012}.  In terms of expressiveness,
these are equivalent to the classical reversal-bounded counter
automata~\cite{Ibarra1978}, but their syntax makes them convenient for our
purposes.

Let $V$ be a finite set of variables.  A \emph{counter automaton over $V$} is
a tuple $\cA=(Q,A,C,E,q_0,F)$, where $Q$ is a finite set of \emph{states}, $A$
is the input alphabet, $C$ is a set of \emph{counters}, 
\[ E\subseteq Q\times (A\cup \{\eword\})^V \times \N^C \times Q \] 
is the finite set of \emph{edges}, $q_0\in Q$ is the
\emph{initial state}, and $F$ is a finite set of pairs
$(q,\varphi)$, where $q\in Q$ and $\varphi$ is an existential Presburger
formula with free variables in $C$.  A \emph{configuration of $\cA$} is a tuple
$(q,\alpha,\mu)$, where $q\in Q$, $\alpha\in (A^*)^V$, $\mu\in \N^C$. The step
relation is defined as follows. We have  
$(q,\alpha,\mu)\autstep[\cA] (q',\alpha',\mu')$
iff there is an edge $(q,\beta,\nu,q')\in E$ such that $\alpha'=\alpha\beta$ and $\mu'=\mu+\nu$. 
A counter automaton accepts a set of assignments, namely
\begin{align*}
\langof{\cA}=\{\alpha\in(A^*)^V \mid \exists (q,\varphi)\in F\colon 
  (q_0,\eword,0)\autsteps[\cA] (q,\alpha,\mu), ~\mu\models\varphi \}\:.
\end{align*}

We call a subset $R\subseteq (A^*)^V$ a \emph{counter relation} if there is a
counter automaton $\cA$ with $R=\langof{\cA}$. If $|V|=1$, say $V=\{x\}$, then
$\cA$ defines a subset of $A^*$, namely the language $\{w\in A^* \mid (x\mapsto w)\in
\langof{\cA}\}$.  Languages of this form are called \emph{counter languages}.

Suppose $V_0,V_1$ are sets of variables with $|V_0\cap V_1|\le 1$. Let
$\cA_i=(Q_i,A,C_i,E_i,q_{0,i},F_i)$ be a counter automaton over $V_i$ for
$i=0,1$ such that $C_0\cap C_1=\emptyset$. Then a simple product construction
yields a counter automaton
$\cA_0\otimes \cA_1=(Q_0\times Q_1,A,C_0\cup C_1,E,(q_{0,0},q_{0,1}),F)$ over $V_0\cup V_1$
such that $((p_0,p_1), \alpha, \mu)\autsteps[\cA_0\otimes\cA_1] ((p'_0,p'_1),\alpha',\mu')$ iff 
\begin{align*}
(p_i,\alpha|_{V_i}, \mu|_{C_i}) \autsteps_{\cA_i} (p'_i,\alpha'|_{V_i},\mu'|_{C_i})
\end{align*}
for $i=0,1$ and 
\[
F=\{((p_0,p_1), \varphi_0\wedge \varphi_1) \mid \text{$(p_i,\varphi_i)\in F_i$ for $i=0,1$}\}.
\]

\begin{proposition}\label{complexity:treeshaped:counterrelation}
Given a tree-shaped constraint system $T$, one can construct in exponential time
a counter automaton $\cA$ with $\langof{\cA}=\assignments{T}$.
\end{proposition}
\begin{proof}
First, observe that it suffices to consider the case where every constraint in
$T$ involves two variables: The other constraints have the form $x=w$ for some
$w\in A^*$ or $x\in(a_1^*\cdots a_n^*)^\ell$ for some $\ell\in\N$ and can easily
be imposed afterwards in the counter automaton.

We construct the automaton inductively. The statement is trivial if $T$
involves only one variable, so assume $|V|\ge 2$ from now on.

Since $\Gamma(T)$ is a forest, it contains a vertex $x\in V$ with at most one
neighbor. Let $T'$ be the constraint system obtained from $T$ by removing all
constraints involving $x$ and suppose we have already constructed a counter
automaton $\cA'$ with $\langof{\cA'}=\assignments{T'}$. 

Now if $x$ has no neighbor, it is easy to construct the automaton for $T$.  So
suppose $x$ has a unique neighbor $y$.  Then, the additional constraints
imposed in $T$ are all $(x,y)$-constraints. Let $T''$ be the set of all
$(x,y)$-constraints in $T$. Now note that if $\cA''$ is a counter automaton
with $\langof{\cA''}=\assignments{T''}$, then we have 
\[ \langof{\cA'\otimes\cA''} = \assignments{T'\cup T''}=\assignments{T}. \]
Therefore, it suffices to construct in polynomial time a counter automaton $\cA''$
with $\langof{\cA''}=\assignments{T''}$.

Observe that any set of $(x,y)$-constraints can be written as a disjunction of
one of the following constraints:
\begin{align*} \text{(i)~$x=y$} && \text{(ii)~$x\strictsubword y$} && \text{(iii)~$y\strictsubword x$} && \text{(iv) $x\bot y$} 
\end{align*}
Since it is easy to construct a counter automaton for the union of two
relations accepted by counter automata, it suffices to construct a counter
automaton for the set of solutions to each of the constraints
(i)---(iv).
This is obvious in all cases but the last. In that last case, one can notice
that $x\bot y$ holds if either
\begin{enumerate*}
\item $|x|<|y|$ and $x\not\subword y$ or
\item $|y|<|x|$ and $y\not\subword x$ or \item $|x|=|y|$ and $x\ne y$.
\end{enumerate*}
Note that each of these cases is easily realized in a counter automaton since
we can use the counters to guarantee the length constraints. Moreover, the
resulting counter automaton can clearly be constructed in polynomial time,
which completes the proof.
\end{proof}

We can now prove \cref{complexity:two-variables:counterrelation} by taking the
constraint system provided by \cref{complexity:constraint-decomposition} and
construct a counter automaton just for $T$ using
\cref{complexity:treeshaped:counterrelation}.  Then, we can impose the
constraints in $S$ by using additional counters. Note that since all variables
in $S$ are alternation bounded, we can store these words, in the form of their
occurring exponents, in counters. We can then install the polynomial-size
Presburger formulas from \cref{subwords-in-presburger} in the counter automaton
to impose the binary constraints required by $S$. This results in an
exponential size counter automaton that accepts the satisfying assignments of
$\varphi$.  The $\NEXP$ upper bound then follows from the fact that
non-emptiness for counter automata is $\NP$-complete.

Since counter automata are only a slight extension of reversal-bounded
counter automata, the following is well-known.
\begin{lemma}\label{complexity:counter:emptiness}
The non-emptiness problem for counter automata is $\NP$-complete.
\end{lemma}
\begin{proof}
Given a counter automaton $\cA=(Q,A,C,E,q_0,F)$, and a state $q\in Q$, we
can construct an existential Presburger formula $\theta_q$ with a free variable
for each edge in $\cA$ that is satisfied for an assignment $\nu\in\N^E$ iff
there is a run from $q_0$ to $q$ where each edge $e\in E$ occurs exactly
$\nu(e)$ times. This is just the fact that we can construct in polynomial time
an existential Presburger formula for the Parikh image of a finite
automaton~\cite{VermaSS05}.

For each edge $e=(p,\alpha,\mu,p')$, define $\mu_e=\mu$.  Then the formula
\[ \bigvee_{(q,\varphi)\in F} \theta_q\wedge \bigwedge_{c\in C} c = \mu_e(c)\cdot e \wedge \varphi \]
expresses precisely that there is an accepting run. The fact that the
satisfiability problem for existential Presburger arithmetic is
$\NP$-complete~\cite{Oppen1978} now gives us the upper bound as well as the
lower bound.
\end{proof}

We are now ready to prove \cref{complexity:two-variables:counterrelation}.
\begin{proof}
First we use \cref{complexity:constraint-decomposition} to turn the formula
$\varphi$ into constraint systems $T$ and $S$ such that $T$ is tree-shaped, $S$
is alternation bounded, and $\assignments[\freevar{\varphi}]{T\cup
S}=\assignments{\varphi}$. Then we use
\cref{complexity:treeshaped:counterrelation} to to obtain in exponential time a
counter automaton $\cA$ with $\langof{\cA}=\assignments{T}$.

Let $U=\{x_1,\ldots,x_m\}$ be the set of variables occurring in $S$.  It
remains to impose the constraints in $S$.  We do this by first building the
product with one automaton $\cA_i$ for each $x_i$.  This automaton imposes the
alternation constraint on $x_i$ and stores the word read into $x_i$ in a set of
counters.  Note that this is possible because the word is alternation bounded.
After taking the product with all these automata, we impose the remaining constraints
from $S$ (which are binary constraints or of the form $x=w$ with $x\in U$, $w\in
A^*$) by adding existential Presburger formulas that express subword
constraints as provided by \cref{subwords-in-presburger}.

We may clearly assume that whenever there is a variable $x$, an alternation
constraint $x\in(a_1^*\cdots a_n^*)^\ell$, and a constraint $x=w$, then
$w\in(a_1^*\cdots a_n^*)^\ell$: Otherwise, the system is not satisfiable and
clearly has an equivalent counter automaton.

Let $A=\{a_1,\ldots,a_n\}$.  Since $S$ is alternation bounded, $S$ contains an
alternation constraint $x_i\in (a_1^*\cdots a_n^*)^{\ell_i}$ for each
$i\in[1,m]$. Let $\ell$ be the maximum of all these $\ell_i$.  We will use the
counter variables $c_{i,j,k}$ for each $i\in [1,m]$, $j\in [1,\ell]$, and $k\in
[1,n]$. We set up the automaton $\cA_i$ over $V_i=\{x_i\}$ such that it has an
initial state $q_0$, a state $q_1$, and satisfies
\[ (q_0,\eword,0) \autsteps[\cA_i] (q_1,\alpha,\mu) \]
if and only if $\alpha$ maps $x_i$ to the word
\begin{equation}
a_1^{\mu(c_{i,1,1})} \cdots a_n^{\mu(c_{i,1,n})} \cdots a_1^{\mu(c_{i,\ell_i,1})} \cdots a_n^{\mu(c_{i,\ell_i,n})}.\label{complexity:word-in-counters}
\end{equation}
This can clearly be done with $n\cdot\ell$ states. Moreover, let
$F_i=\{(q_1,\top)\}$.  Note that $\cA_i$ has the counters $c_{i,j,k}$ even for
$j>\ell_i$ although it never adds to them.  The reason we have the variables
$c_{i,j,k}$ for $j>\ell_i$ is that this way, the formulas from
\cref{subwords-in-presburger} are applicable.

Note that since each $V_i$ is a singleton, the automaton
$\cB=\cA\times\cA_1\otimes\cdots\otimes\cA_m$ is defined. It satisfies
$\langof{\cB}=\assignments{S'}$, where $S'$ is the set of alternation
constraints in $S$.

It remains to impose the remaining constraints from $S$, namely the
binary constraints and those of the form $x=w$ with $x\in U$, $w\in A^*$. Let
$R\subseteq S$ be the set of these remaining constraints. For each $i$ and for
$\mu\in (A^*)^V$, let $w_{\mu,i}$ be the word in
\cref{complexity:word-in-counters}. According to \cref{subwords-in-presburger},
for each constraint $r\in R$, we can construct a polynomial size existential
Presburger formula $\kappa_r$ such that $\mu \models \kappa_r$ if and only if
the constraint is satisfied for the assignment $\alpha$ with
$\alpha(x_i)=w_{\mu,i}$. Moreover, let $\kappa=\bigwedge_{r\in R} \kappa_r$.

Suppose $\cB=(Q,A,C,E,q_0,F)$.
In the last step, we construct the counter automaton $\cB'=(Q,A,C,E,q_0,F')$, where
\[ F'=\{(q,\psi\wedge\kappa) \mid (q,\psi)\in F \}. \]
Now $\cB'$ clearly satisfies $\langof{\cB'}=\assignments{T\cup S}$.  Thus, if
we obtain $\cB''$ from $\cB'$ by projecting the input to those variables that
occur freely in $\varphi$, then
$\langof{\cB''}=\assignments[\freevar{\varphi}]{T\cup S}$. Moreover, $\cB''$ can
clearly be constructed in exponential time.

The membership of the truth problem in $\NEXP$ follows from the fact that
emptiness of counter automata is in $\NP$ (\cref{complexity:counter:emptiness}).
\end{proof}

\section{Expressiveness}\label{expressiveness}

In this section, we shed some light on which predicates or languages are
definable in our fragments $\frag{i}{j}$.

\subsection{Expressiveness of the $\frag{1}{0}$ fragment}
A language $L$ definable in $\frag{1}{0}$ always satisfies
$L\subseteq (a_1^*\cdots a_n^*)^\ell$ for some $\ell\in\N$. Hence, it can be
described by the set of vectors that contain the occurring exponents.  As can
be derived from results in \cref{complexity}, these sets are always semilinear.
In this section, we provide a decidable characterization of the semilinear sets
that are expressible in this way. Stating the characterization requires some
terminology.

Let $V$ be a set of variables. By $\N^V$, we denote the set of mappings
$V\to\N$.  By a \emph{partition of $V$}, we mean a set $P=\{V_1,\ldots,V_n\}$
of subsets $V_1,\ldots,V_n\subseteq V$ such that $V_i\cap V_j=\emptyset$ for
$i\ne j$ and $V_1\cup \cdots\cup V_n=V$. If $U\cap V=\emptyset$ and $\alpha\in
\N^U$, $\beta\in\N^V$, we write $\alpha\times\beta$ for the map
$\gamma\in\N^{U\cup V}$ such that $\gamma|_U=\alpha$ and $\gamma|_V=\beta$.
Furthermore, if $S\subseteq \N^U$, $T\subseteq \N^V$, then $S\times
T=\{\alpha\times\beta \mid \alpha\in S,~ \beta\in T\}$.
A semilinear set $S\subseteq \N^V$ is \emph{$P$-compatible} if it has
a semilinear representation where each occurring period vector belongs
to $\N^{V_i}$ for some $i\in[1,n]$.

\begin{theorem}\label{expressiveness:compatible}
Suppose $L\subseteq (a_1^*\cdots a_n^*)^\ell$.  Let $V=\{x_{i,j} \mid
i\in[1,\ell],~j\in[1,n]\}$ and consider the partition $P=\{V_1,\ldots,V_n\}$
where $V_j=\{x_{i,j}  \mid i\in[1,\ell]\}$ for $j\in[1,n]$. The language $L$ is
definable in $\frag{1}{0}$ if, and only if, the set
\[ \{\alpha\in\N^V \mid a_1^{\alpha(x_{1,1})}\cdots a_n^{\alpha(x_{1,n})}\cdots a_1^{\alpha(x_{\ell,1})}\cdots a_n^{\alpha(x_{\ell,n})}\in L \} \]
is a $P$-compatible semilinear set.
\end{theorem}
For example, this means we can define $\{a^nba^n \mid n\in\N\}$, but not
$\{a^nb^n\mid n\in\N\}$: A semilinear representation for the latter requires
a period that produces both $a$'s and $b$'s.

The proof of \cref{expressiveness:compatible} employs a characterization of
$P$-compatible sets in terms of Presburger arithmetic.  Let $V$ be a set of
variables and $\varphi$ be a Presburger formula whose variables are in $V$.
Let $P=\{V_1,\ldots,V_n\}$ be a partition of $V$.  We say $\varphi$ is
\emph{$P$-compatible} if there is a set of variables $V'\supseteq V$ and a
partition $P'=\{V'_1,\ldots,V'_n\}$ of $V'$ such that 
\begin{enumerate}
\item $V_j\subseteq V'_j$ for each $j\in[1,n]$ and
\item in each literal in $\varphi$, all variables belong to the same set $V'_j$
for some $j\in[1,n]$.
\end{enumerate}

The following is a simple observation.
\begin{theorem}\label{expressiveness:partition}
Let $P=\{V_1,\ldots,V_n\}$ be a partition of $V$. For sets $S\subseteq \N^V$,
the following conditions are equivalent:
\begin{enumerate}
\item\label{partition:semilinear} $S$ is a $P$-compatible semilinear set.
\item\label{partition:formula} $S=\llbracket \varphi\rrbracket$ for some $P$-compatible existential
Presburger formula $\varphi$.
\item\label{partition:cartesian} $S$ is a finite union of sets of the form $A_1\times\cdots\times A_n$
where each $A_j$ is a semilinear subset of $\N^{V_j}$.
\end{enumerate}
\end{theorem}
\begin{proof}
The directions \direction{partition:cartesian}{partition:semilinear} and
\direction{partition:semilinear}{partition:formula} are easy to see, so we show
\direction{partition:formula}{partition:cartesian}.

If a set satisfies the condition of \ref{partition:cartesian}, then projecting
to a subset of the coordinates yields again a set of this form. Therefore, it
suffices to consider the case where in $\varphi$, there are no quantifiers.
Now, bring $\varphi$ into disjunctive normal form. Since each literal in
$\varphi$ only mentions variables from $V_j$ for some $j\in[1,n]$, we can sort
the literals of each co-clause of the DNF according to the subset $V_j$ they
mention. Hence, we arrive at the form
\[ \varphi \equiv \bigvee_{i=1}^k \bigwedge_{j=1}^n \varphi_{i,j}, \]
where $\varphi_{i,j}$ only mentions variables from $V_j$. This implies
\[ \llbracket\varphi\rrbracket = \bigcup_{i=1}^k \llbracket \varphi_{1,j}\rrbracket \times \cdots\times\llbracket \varphi_{n,j}\rrbracket, \]
which is the form required in \ref{partition:cartesian}.
\end{proof}

We are now ready to prove \cref{expressiveness:compatible}.
\begin{proof}
If $L$ is definable in $\frag{1}{0}$, we can write down a Presburger formula
that defines $S$.  Here, in order to express the subword ordering (and its
negation), we use the formulas from \cref{subwords-in-presburger}.  Observe
that these formulas are $P$-compatible. This means that $S$ is $P$-compatible.

For the converse, suppose $S$ is $P$-compatible. According to
\cref{expressiveness:partition}\ref{partition:semilinear}, $S$ is defined by a
$P$-compatible existential Presburger formula $\varphi$. Hence, $\varphi$ has
free variables $V=\{x_{i,j} \mid i\in [1,\ell], j\in[1,n]\}$ and uses variables
$V'\supseteq V$ that are partitioned as $V'=\biguplus_{j=1}^n V'_j$ so that in
each literal, all occurring variables belong to the same $V'_j$.

In the first step, we turn $\varphi$ into a $\frag{1}{0}$ formula
$\bar{\varphi}$ with the same number of free variables.  For each $x\in V'$, we
take a fresh variable $\bar{x}$, which will hold words in $a_j^*$. More
precisely, we have our new variables $\bar{V}=\{\bar{x} \mid x\in V'\}$ and a
mapping $\iota\colon \N^{V'}\to (A^*)^{\bar{V}}$ defined by
$\iota(\alpha)(\bar{x})=a_j^{\alpha(x)}$, where $j$ is the unique index with
$x\in V'_j$. We want to construct $\bar{\varphi}$ so that
$\assignments{\bar{\varphi}}=\iota(\assignments{\varphi})$.

We obtain $\bar{\varphi}$ from $\varphi$ as follows. For each literal $x=y+z$,
we know that there is a $j\in[1,n]$ with $x,y,z\in V'_i$, so we can replace
the literal with $|\bar{x}|_{a_j}=|\bar{y}|_{a_j}+|\bar{z}|_{a_j}$, which is
expressible in $\frag{1}{0}$ according to \cref{undec2:add} in the proof of
\cref{subword:undecidable} (note that in this case, we actually are in
$\frag{1}{0}$ because the variables $\bar{x}$, $\bar{y}$, and $\bar{z}$ range
over $a_j^*$ and are thus alternation bounded). Since we can clearly also
express $|\bar{x}|_{a_j}\ne |\bar{y}|_{a_j}$ in $\frag{1}{0}$, we use this to
implement literals $x\ne y$ with $x,y\in V'_j$.  Literals of the form $x=k$
with $k\in\N$ and $x\in V'_j$ can just be replaced by $\bar{x}=a_j^k$.  Then we
clearly have $\assignments{\bar{\varphi}}=\iota(\assignments{\varphi})$.

In the second step, we construct the words
\[ a_1^{\alpha(x_{1,1})}\cdots a_n^{\alpha(x_{1,n})}\cdots a_1^{\alpha(x_{\ell,1})}\cdots a_n^{\alpha(x_{\ell,n})} \]
for $\alpha\in\assignments{\varphi}$. This is possible thanks to
\cref{undec2:conc:unary} of the proof of \cref{subword:undecidable}. We can
express
\[ u=\bar{x}_{1,1}\cdots \bar{x}_{1,n}\cdots \bar{x}_{\ell,1}\cdots\bar{x}_{\ell,n} \]
by applying \cref{undec2:conc:unary} exactly $\ell\cdot n-1$ times, once to
append each $x_{i,j}$ to the word defined so far, using $\ell\cdot n-1$
additional variables. Of course, all these variables can be restricted to
$(a_1^*\cdots a_n^*)^\ell$, which means the resulting formula belongs to
$\frag{1}{0}$. Moreover, it clearly defines $L$.
\end{proof}

Our characterization of $\frag{1}{0}$ is decidable.
We use a technique from~\cite{GinsburgSpanier1966a}, where it is shown that
recognizability is decidable for semilinear sets. The idea is to characterize
$P$-compatibility as the finiteness of the index of certain equivalence
relations, which can be expressed in Presburger arithmetic.  
\begin{theorem}\label{expressiveness:compatibility:decidable}
Given a semilinear subset $S\subseteq\N^V$ and a partition $P$ of $V$, it is
decidable whether $S$ is $P$-compatible.
\end{theorem}
\begin{proof}
For $\alpha\in\N^{V_i}$ and $\gamma\in\N^V$, we write
$\gamma[i/\alpha]$ to be the element of $\N^V$ with
\[ \gamma[i/\alpha](v)=\begin{cases} \alpha(v) & \text{if $v\in V_i$,} \\ \gamma(v) & \text{otherwise.}\end{cases} \]
For $\alpha,\beta\in\N^{V_i}$, we write $\alpha\sim_i\beta$ if for every
$\gamma\in\N^V$, we have $\gamma[i/\alpha]\in S$ if and only if
$\gamma[i/\beta]\in S$. Moreover, for $\gamma\in\N^V$, we will use the norm
$\|\cdot\|$ as defined by $\|\gamma\|=\sum_{v\in V} \gamma(v)$. We claim that
$S$ is $P$-compatible if and only if
\begin{equation}
\exists k\in\N\colon \bigwedge_{i=1}^n \forall\alpha\in\N^{V_i} \colon \exists \beta\in\N^{V_i}\colon \|\beta\|\le k,~\alpha\sim_i\beta.\label{charact:dec:compatible}
\end{equation}
Suppose \cref{charact:dec:compatible} holds. For each $\beta\in\N^{V_i}$, we define
$S_{i,\beta}=\{\alpha\in\N^{V_i} \mid \alpha\sim_i\beta\}$.
Then $S_{i,\beta}\subseteq\N^{V_i}$ is semilinear and we have
\[ S = \bigcup_{\beta_1\in\N^{V_1},\|\beta_1\|\le k}\cdots \bigcup_{\beta_n\in\N^{V_n},\|\beta_n\|\le k} S_{1,\beta_1}\times\cdots\times S_{n,\beta_n}. \]
Hence, $S$ is $P$-compatible.

Now assume $S$ is $P$-compatible. Then we can write $S=\bigcup_{j=1}^\ell
A_{j,1}\times \cdots\times A_{j,n}$, where each $A_{j,i}\subseteq\N^{V_i}$ is
semilinear. For each $i\in[1,n]$, consider the function $\kappa_i\colon\N^{V_i}\to \Powerset{\{1,\ldots,\ell\}}$ with
\[\kappa_i(\alpha)=\{j\in[1,\ell] \mid \alpha\in A_{j,i} \}. \]
Observe that if $\kappa_i(\alpha)=\kappa_i(\beta)$, then $\alpha\sim_i\beta$.
Since $\kappa_i$ has a finite codomain, this means the equivalence relation
$\sim_i$ on $\N^{V_i}$ has finite index. This immediately implies
\cref{charact:dec:compatible}.

Since we can clearly formulate the condition \cref{charact:dec:compatible} in
Presburger arithmetic, $P$-compatibility is decidable.
\end{proof}

In fact, it is not
hard to see that if $P$ consists only of singletons, a semilinear set is
$P$-compatible iff it is recognizable.  Hence,
\cref{expressiveness:compatibility:decidable} generalizes the decidability of
recognizability.
Let $M$ be a monoid. A subset $S\subseteq M$ is called \emph{recognizable} if
there is a finite monoid $F$ and a morphism $\varphi\colon M\to F$ such that
$S=\varphi^{-1}(\varphi(S))$.
\begin{theorem}\label{singletons:recognizable}
Suppose $P$ consists only of singletons. Then $S\subseteq\N^V$ is
$P$-compatible if and only if it is recognizable.
\end{theorem}
\begin{proof}
Mezei's Theorem~\cite{berstel79} states that if $M_1,\ldots,M_n$ are monoids,
then a subset of $M_1\times\cdots\times M_n$ is recognizable if and only if it
is a finite union of sets $S_1\times \cdots\times S_n$ such that $S_i\subseteq
M_i$ is recognizable for $i\in\{1,\ldots,n\}$.

Combined with the fact that a subset of $\N$ is semilinear if and only if it is
recognizable, the condition \ref{partition:cartesian} of
\cref{expressiveness:partition} yields the result.
\end{proof}

\subsection{Expressiveness of $\frag{1}{0}$ vs. $\frag{1}{1}$}
It is obvious that
$\frag{1}{1}$ is strictly more expressive than $\frag{1}{0}$, because it permits the definition of
languages with unbounded alternations, such as $\{a,b\}^*$. But is this the
only difference between the two fragments? In other words: Restricted to alternation
bounded languages, is $\frag{1}{1}$ more expressive? The answer is no.
\begin{theorem}
If $L\subseteq (a_1^*\cdots a_n^*)^\ell$, then $L$ is definable in
$\frag{1}{1}$ if and only if it is definable in $\frag{1}{0}$.
\end{theorem}
\begin{proof}
  Let $\varphi$ be a $\frag{1}{1}$ formula where the free variable is
  alternation bounded and the variable $t$ is not alternation
  bounded. We can transform $\varphi$ into a disjunction
  $\bigvee_{i=1}^k \varphi_i$, where each $\varphi_i$ belongs to
  $\frag{1}{1}$ and consists of a block of existential quantifiers
  followed by a conjunction of literals. Then, the proof of
  \cref{thm:sat-1-1} yields for each $\varphi_i$ a (polynomial) bound
  $p_i$ so that if we replace the quantifier $\exists t$ in
  $\varphi_i$ by $\exists t\in (a_1^*\cdots a_n^*)^{p_i}$, the
  resulting $\frag{1}{0}$ formula is equivalent.
\end{proof}

This allows us to reason beyond alternation bounded languages. We have seen in
the proof of \cref{subword:undecidable} that one can express ``$|u|_a=|v|_b$''
in $\frag{1}{3}$, which required significantly more steps, and two more
alternation unbounded variables, than the ostensibly similar ``$|u|_a=|v|_a$''.
This raises the question: Can we define the former in $\fragf{1}{1}$? We cannot:
\begin{corollary}
The predicate ``$|u|_a=|u|_b$'' is not definable in $\frag{1}{1}$.
\end{corollary}
\begin{proof}
Otherwise, we could define the set $\{a^nb^n \mid
n\ge 0\}$ in $\frag{1}{1}$, hence in $\frag{1}{0}$, contradicting \cref{expressiveness:compatible}.
\end{proof}

\subsection{Expressiveness of $\frag{1}{2}$ vs. $\frag{1}{3}$}
We have seen in \cref{subword:undecidable} that $\frag{1}{3}$ can express all
recursively enumerable unary languages. Moreover,
\cref{complexity:two-variables:counterrelation} tells us that the languages
definable in $\frag{1}{2}$ are always counter languages.

How do the fragments compare with respect to natural (binary) predicates on
words. We already know from \cref{undec2:prefix} in \cref{subword:undecidable}
that over two letters, the prefix relation is expressible in $\frag{1}{3}$.  Note that the
following \lcnamecref{expressiveness:prefix} does not follow directly from the
fact that for any $\frag{1}{2}$ formula $\varphi$, the set
$\assignments{\varphi}$ is a counter relation, as shown in \cref{complexity}.
This is because the prefix relation is a counter relation (and even rational). 
\begin{theorem}\label{expressiveness:prefix}
In $\frag{1}{2}$, ``$u$ is a prefix of $v$'' is not expressible.
\end{theorem}
\begin{proof}
Suppose the prefix relation were expressible using a $\frag{1}{2}$ formula
$\varphi$. Then, by reversing all constants in $\varphi$, we obtain a formula
expressing the suffix relation.  Let $\prefix$ denote the prefix relation and
$\suffix$ the suffix relation. We can now express
\[ \exists v\in\{a,b\}^*\colon v\prefix u\wedge v\suffix u \wedge |u|_a=2\cdot |v|_a \wedge |u|_b=2\cdot |v|_b, \]
which is equivalent to $u\in S$, where $S=\{vv \mid v\in\{a,b\}^*\}$.  Note
that $|u|_a=2\cdot |v|_a$ can be expressed by using $|u|_a=|v|_a+\cdot |v|_a$,
which can be done in $\fragf{1}{0}$ according to \cref{undec2:add} in
\cref{subword:undecidable}.

However, $S$ is not a counter language. This is due to the fact that the class
of recursively enumerable languages is the smallest language class that
contains $S$, is closed under rational transductions, union, and intersection
(this can be shown as in the case of the set of
palindromes~\cite{BakerBook1974}).  However, the class of counter languages
also has these closure properties and is properly contained in the recursively
enumerable languages. Hence, $S$ is indeed not a counter language.  This is in
contradiction with the fact that for any $\frag{1}{2}$ formula  $\varphi$,
the set $\assignments{\varphi}$ is a counter relation, as shown in \cref{complexity}.
\end{proof}

\section{Conclusion}
\label{sec-concl}

We have shown that the $\Sigma_1$ theory of the subword ordering is undecidable
(already for two letters), if all words are available as constants.  This
implies that the $\Sigma_2$ theory is undecidable already for two letters, even
without constants.

In order to shed light on decidable fragments of first-order logic over
the structure $\swstructconst{A}$, we introduced the fragments $\frag{i}{j}$. We have completely
settled their decidability status. In terms of complexity, the only open case
is the $\frag{1}{2}$ fragment.  We have an $\NP$ lower bound and an $\NEXP$
upper bound. 

This aligns with the situation for expressiveness. We have a decidable
characterization for the expressiveness of $\frag{1}{0}$ and, obvious
exceptions aside, $\frag{1}{1}$ is as expressive as $\frag{1}{0}$. However, we
do not know whether $\frag{1}{1}$ and $\frag{1}{2}$ differ significantly: Of
course, $\frag{1}{2}$ can have two alternation unbounded free variables, but it
is conceivable that $\frag{1}{1}$ and $\frag{1}{2}$ define the same languages.

\printbibliography

\end{document}